\documentclass[%
 reprint,
 superscriptaddress,
nofootinbib,
 amsmath,amssymb,
 aps,
 prc,
]{revtex4-1}

\usepackage[colorlinks=true,allcolors=blue]{hyperref} % blue cross-refs
\usepackage{graphicx}% Include figure files
\usepackage{dcolumn}% Align table columns on decimal point
\usepackage{bm}% bold math
\usepackage{physics} % Physics symbols
\renewcommand{\vec}[1]{\boldsymbol{#1}} % Vectors are bold
\newcolumntype{d}[1]{D{.}{.}{#1}} % Specify float format with column d{x.x}
\usepackage{xcolor}
\usepackage{soul}
\usepackage{amsthm}
\newtheorem{property}{Property}
\newtheorem{corollary}{Corollary}[property]

\usepackage{ulem}

\newcommand{\ts}{\ensuremath{ \widetilde{\mathcal{S}}}}
\newcommand{\tsp}{\ensuremath{ \widetilde{\mathcal{S}}_{p} }}
\newcommand{\trp}{\ensuremath{ \delta \widetilde{R}_{p} }}
\newcommand{\trr}{\ensuremath{ \delta \widetilde{R} }}
\newcommand{\tlp}{\ensuremath{ \delta \widetilde{L}_{p} }}

\begin{document}

%\preprint{APS/123-QED}

\title{Elucidating the finite temperature quasiparticle random phase approximation}

\author{E. M. Ney}
\email[]{evan.ney@unc.edu}
\affiliation{Department of Physics and Astronomy, CB 3255, University of North
Carolina, Chapel Hill, North Carolina 27599-3255, USA\looseness=-1}

\author{A. Ravli\'c}
\email[]{aravlic@phy.hr}
\affiliation{Department of Physics, Faculty of Science, University of Zagreb, Bijeni\v{c}ka c. 32, 10000 Zagreb, Croatia}

\author{J. Engel}
\email[]{engelj@unc.edu}
\affiliation{Department of Physics and Astronomy, CB 3255, University of North
Carolina, Chapel Hill, North Carolina 27599-3255, USA\looseness=-1}

\author{N. Paar}
\email[]{npaar@phy.hr}
\affiliation{Department of Physics, Faculty of Science, University of Zagreb, Bijeni\v{c}ka c. 32, 10000 Zagreb, Croatia}

\date{\today}

\begin{abstract}
In numerous astrophysical scenarios, such as core-collapse supernovae and neutron star mergers, as in well as heavy-ion collision experiments, transitions between thermally populated nuclear excited states have been shown to play an important role. Due to its simplicity and excellent extrapolation ability, the finite-temperature quasiparticle random phase approximation (FT-QRPA) presents itself as an efficient method to study the properties of hot nuclei.  The statistical ensembles in the FT-QRPA make the theory much richer than its zero-temperature counterpart, but also obscure the meaning of various physical quantities. In this work, we clarify several aspects of the FT-QRPA, including notations seen in the literature, and demonstrate how to extract physical quantities from the theory. To exemplify the correct treatment of finite-temperature transitions, we place special emphasis on the charge-exchange transitions described within the proton-neutron FT-QRPA (FT-PNQRPA). With the FT-PNQRPA built on the nuclear energy-density functional theory, we obtain solutions using a relativistic matrix approach and also the non-relativistic finite amplitude method. We show that the Ikeda sum rule is fulfilled with the proper treatment of de-excitations from thermally populated excited states. Additionally, we demonstrate the impact of these transitions on stellar electron capture (EC) rates in ${}^{58,78}$Ni. While their inclusion does not influence the EC rates in ${}^{58}$Ni, the rates in ${}^{78}$Ni are dominated by de-excitations for temperatures $T > 0.5$~MeV. In systems with a large negative $Q$-value, the inclusion of de-excitations within the FT-QRPA is necessary for a complete description of reaction rates at finite temperature.
\end{abstract}

\maketitle

\section{Introduction}

Nuclear decays from thermally populated excited states are ubiquitous, occurring in settings from heavy-ion fusion reactions in laboratories to extreme astrophysical environments such as supernovae and neutron star mergers~\cite{JANKA200738,Egido_1993,Thielemann_2017,KAJINO2019109,Shlomo_2004,Gross_1990}. The temperature-dependent mean field theory is a reasonable method for describing thermally equilibrated nuclear systems because its numerical efficiency scales slowly with system size, and symmetry-unrestricted mean field calculations with density-dependent effective interactions provide an accurate, microscopic description of nuclear structure and dynamics. Moreover, nuclear density functional theory (DFT) based on energy-density functionals (EDFs) with no connection to an underlying interaction can effectively incorporate correlations while still using the simple mean-field description~\cite{NIKSIC2011519,Schunck_book,ROCAMAZA201896,ROCAMAZA201896}. 

% construction of mean-field theories
One common approach to construct finite-temperature mean-field theories is based on statistical ensembles~\cite{Goodman1981,Sano_1963,Egido_1993}. Expectation values of an operator at finite-temperature are taken with respect to the mean-field statistical density operator, which implies the Boltzmann weighted summation over quasiparticle states, assuming either the canonical or grand-canonical ensemble. In the case of superfluid nuclei, the finite-temperature Hartree-Fock-Bogoliubov (FT-HFB) equations were derived in Ref.~\cite{Goodman1981,Egido_1993}, having the same form as the zero-temperature HFB equations but involving temperature-dependent densities~\cite{PhysRevC.34.1942,Egido_1993,Reiss_1999,PhysRevC.68.034327,PhysRevC.88.034308}.
% The interaction Hamiltonian $\hat{H}$ is approximated with a Hamiltonian of independent quasiparticles which are described by a Fermi-Dirac distribution. The expressions for expectation values are taken by calculating the thermal averages using the generalized Wick theorem. In practical calculations, the FT-HFB implementation can be easily constructed from the zero-temperature HFB, by changing appropriate definitions of the q.p. density $\rho$ and the pairing density $\kappa$. Thus, most numerical implementations of the FT-H(F)B theory describing nuclear properties are based on the statistical formulation
On the other hand, instead of using the statistical ensembles one can use the notion of a thermal vacuum within the thermo-field dynamics (TFD)~\cite{Takahashi_1996,Umezawa1982}. The thermal vacuum within TFD is defined such that it yields expectation values equivalent to the corresponding statistical ensemble thermal averages. To this aim, the dimension of the original Hamiltonian is doubled by introducing the so-called fictitious system operators. The thermal vacuum is then constructed by a Bogoliubov transformation between original and fictitious system operators.

% construction of FT-QRPA
The finite-temperature quasiparticle random-phase approximation (FT-QRPA) is the linearized time-dependent FT-HFB theory. Its zero-temperature limit is well known for its success in describing small-amplitude collective motion~\cite{Paar_2007, ROCAMAZA201896}. The FT-QRPA, much like FT-HFB, can be built either using statistical ensembles~\cite{RING1984261, Sommermann1983} or by the TFD formalism (also known as the thermal QRPA or TQRPA)~\cite{Civitarese1992}. The latter is achieved by introducing the phonon creation operator containing combinations of both the original and fictitious quasiparticle operators and diagonalizing the corresponding residual interaction Hamiltonian in that basis. Thus, the transitions between the original one-phonon states are described as excitations while the transitions to fictitious one-phonon states represent de-excitations (transitions with energy below the $Q$-value threshold). The TFD formulation has been applied to study the thermal evolution of the multipole strength functions and weak-interaction rates in numerous works~\cite{PhysRevC.42.1335,PhysRevC.94.015805,PhysRevC.92.045804,PhysRevC.101.025805,PhysRevC.81.015804,Dzhioev2009,PhysRevC.100.025801,Dzhioev2008}.

The FT-QRPA, based on taking the FT-HFB thermal averages, was derived in Ref.~\cite{Sommermann1983}, with a plethora of implementations employing various model interactions. Especially significant are the early implementations based on schematic models~\cite{RING1984261,BARRANCO1985445,PhysRevC.63.034323,Civitarese_1999,PhysRevC.62.054318} and recent self-consistent implementations based on nuclear EDFs, both non-relativistic~\cite{KHAN2004311,PhysRevC.80.065808,PhysRevC.80.055801,PhysRevC.96.024303,Yuksel2019} and relativistic~\cite{PhysRevC.83.045807,NIU2009315,PhysRevC.101.044305,PhysRevC.104.064302}. Furthermore, in the recent works of Refs.~\cite{PhysRevLett.121.082501,PhysRevC.100.024307,Litvinova2019,PhysRevC.103.024326,LITVINOVA2020135134} the finite-temperature nuclear response is formulated beyond the one-loop approximation of the FT-QRPA within the finite-temperature relativistic time-blocking approximation (FT-RTBA). 

The FT-QRPA based on thermal averages is, however, conceptually difficult to interpret. Furthermore, several different notations have been used for the FT-QRPA, and the connection between these notations has not been explicitly clarified.
% The connection between mean field and exact descriptions of statistical ensembles is discussed sporadically across the literature, and is sometimes treated inconsistently. 
Especially confusing is the use of the so-called thermal prefactor which multiplies the overall strength function. It is defined as $(1-e^{-\beta \omega})^{-1}$ where $\omega$ is the excitation energy and $\beta = 1/(k_B T)$, $T$ being the temperature and $k_B$ the Boltzmann constant. This prefactor, which is typically derived by invoking the principle of detailed balance~\cite{Sommermann1983,RING1984261,Chomaz1990,PhysRevC.92.045804}, has not been consistently included in FT-QRPA calculations. In part, this may be attributed to uncertainty as to whether detailed balance applies in specific physical applications, such as weak reactions in stellar environments.
% Previous works that include this prefactor usually derive it using the detailed balance formalism, which relates the strength function for positive excitation energies with the strength function at negative excitation energies \cite{Sommermann1983,RING1984261}. Others just omit this factor altogether [\textcolor{red}{Check if OK}].
The thermal prefactor significantly modifies low-lying strength at finite-temperature, as was exemplified in Refs.~\cite{Chomaz1990, PhysRevC.100.024307}, so its use should be understood and justified. In addition to the prefactor, the identification of transitions in FT-QRPA strength function is also difficult. In a shell model calculation, for example, every transition is calculated explicitly. However, the FT-QRPA includes transitions among phonons in a thermal ensemble of many-quasiparticle states, and it is not immediately clear where specific excitations and de-excitations appear in the strength function. This obfuscates comparisons between strength functions computed with the FT-QRPA and other approaches like the shell model. In this work, we aim to clarify all these aspects of the FT-QRPA and emphasize the proper way to use the theory to study decays in hot nuclei.
% origin of the thermal prefactor by deriving it from the properties of the exact finite-temperature Green's function. Furthermore, we attempt to illuminate the connection between the physical strength function to the FT-QRPA response function, as well as the connection between different theoretical formulations and notations. We interpret the transition strength function that arises from the FT-QRPA and outline how it should be used to study decays in hot nuclei. 

We demonstrate the correctness of our discussion by investigating implications for the temperature evolution of Gamow-Teller (GT) strength in ${}^{58}$Ni as well as the electron capture (EC) rates in ${}^{58,78}$Ni for temperatures up to 2 MeV. These nuclei are found in abundance during the late-stage core-collapse supernovae (CCSNe) evolution~\cite{Sullivan_2015, JANKA200738,Langanke_2021}. Calculations are performed using two different state-of-the-art models, the non-relativistic Skyrme FT-QRPA based on the proton-neutron finite-amplitude method (PNFAM) and the relativistic proton-neutron FT-QRPA (FT-PNRQRPA) employing the D3C${}^*$ interaction (refer to Ref.~\cite{PhysRevC.105.055801} for more details).

This paper is structured as follows: In Sec.~\ref{sec:props} we expound on the notation and properties of the FT-QRPA seen in the literature, supplemented with Appendix~\ref{sec:appa}. Section~\ref{sec:transitions} demonstrates how the FT-QRPA approximates the exact linear response function and explains how to extract physical quantities from it. To illustrate the discussion in Sec.~\ref{sec:transitions}, we calculate Gamow-Teller strength functions using the charge-exchange FT-QRPA in Sec.~\ref{sec:bgt}. We discuss the results in the context of the Ikeda sum rule, which is derived in the present formulation of the FT-QRPA (with additional details in Appendices~\ref{sec:appb} and~\ref{sec:appc}), and show that it is fulfilled if the FT-QRPA strength function is treated correctly. Finally, in Sec.~\ref{sec:ec_rates} we study implications for the temperature evolution of stellar EC rates in ${}^{58,78}$Ni.

\section{Properties of the FT-QRPA}\label{sec:props}
\subsection{Finite-temperature linear response}

The seminal work of Ref.~\cite{Sommermann1983} derived the linear response equations for FT-HFB ensembles~\cite{Goodman1981}. There, the resulting expressions are given in a particular temperature-symmetric form. In this work, however, we focus on a less symmetrical but more general form and discuss how it relates to Ref.~\cite{Sommermann1983}. To see this temperature symmetry more clearly, we keep all temperature-dependent factors separate from temperature-independent factors in our discussion, except when indicated by a tilde.

The FT-HFB linear response equations are derived by linearizing the time-dependent FT-HFB equations. We can write the result from Ref.~\cite{Sommermann1983} compactly as,
\begin{equation}\label{eq:ft_linear_response}
\begin{gathered}
  \Big[ \widetilde{\mathcal{S}} - \omega M \Big] \delta \widetilde{R}(\omega) =  -T \mathcal{F}(\omega)
  \\
  \widetilde{\mathcal{S}} \equiv T\mathcal{H} + \mathcal{E}
  ,\quad
  \delta \widetilde{R} \equiv T \delta R
  \,,
\end{gathered}
\end{equation}
where above matrices and vectors are defined in an extended 4-component supermatrix space whose elements live in an enlarged two-quasiparticle space. The matrix $\mathcal{H}$ represents the residual interaction, where expressions for its sub-matrices appear in Appendix~B of Ref.~\cite{Sommermann1983}, while the quantity $M$ takes the role of the metric in the 4-component supermatrix space. Matrix elements of $T$ and $\mathcal{E}$ depend on the quasiparticle occupations, $f_\alpha = \left[ 1 + \exp \left(E_\alpha/k_B T \right) \right]^{-1},
$ and energies, $E_\alpha$, for two quasiparticles. Using the abbreviations $E^{\pm}_{\alpha \beta} \equiv E_{\alpha} \pm E_{\beta}$ and $f^{\pm}_{\alpha \beta} \equiv f_{\alpha} \pm f_{\beta}$, the matrices in Eq.~\eqref{eq:ft_linear_response} are defined as,
\begin{equation}
\begin{aligned}
    T_{\alpha \beta, \gamma\delta}
    % &=\text{diag}[(f_\beta - f_\alpha), (1 -f_\alpha - f_\beta), (1 -f_\alpha - f_\beta), (f_\beta - f_\alpha)] \delta_{\gamma \delta}
    &=\text{diag}[ f_{\beta \alpha}^-, (1 - f_{\alpha \beta}^+), (1 - f_{\alpha \beta}^+),  f_{\beta \alpha}^-] \delta_{\gamma \delta}
\\
    \mathcal{E}_{\alpha \beta, \gamma\delta}
    % &=\text{diag}[(E_\alpha - E_\beta),  (E_\alpha + E_\beta), (E_\alpha + E_\beta), (E_\alpha - E_\beta)] \delta_{\gamma \delta}
    &=\text{diag}[E_{\alpha \beta}^-,  E_{\alpha \beta}^+, E_{\alpha \beta}^+, E_{\alpha \beta}^-] \delta_{\gamma \delta}
\\
    M_{\alpha\beta,\gamma\delta}
    &= \text{diag}[\phantom{-}1,\phantom{-}1,-1,-1] \delta_{\alpha \beta, \gamma \delta}
\\
    \mathcal{H}_{\alpha \beta, \gamma\delta}
    &= \frac{\partial H_{\alpha \beta}}{\partial R_{\gamma\delta}}
    \,.
\end{aligned}
\end{equation}
Finally, the vectors in Eq.~\eqref{eq:ft_linear_response} are the density response, ${\delta {R}_{\alpha\beta}(\omega) = (P_{\alpha\beta}, X_{\alpha\beta}, Y_{\alpha\beta}, Q_{\alpha\beta})}$, and the external field, $\mathcal{F}_{\alpha\beta}(\omega) = (F^{11}_{\alpha\beta}, F^{20}_{\alpha\beta}, F^{02}_{\alpha\beta}, F^{\bar{11}}_{\alpha\beta})$.

The connection to Ref.~\cite{Sommermann1983} comes from imposing restrictions on the form of $T$ and grouping the temperature-dependent factors in a particular way. Assuming the system is not degenerate, without loss of generality we can construct the two-quasiparticle basis $\{\alpha\beta\}$ in an order such that $E_\alpha > E_\beta$. Then $T$ is positive definite, and we may define powers of $T$ in the sense that $T = T^{1-p} T^p$ for $0 \le p \le 1 \in \mathbb{R}$. By eliminating a factor of $T^{1-p}$ on the left from both sides, we see that Eq.~\eqref{eq:ft_linear_response} is just one of an infinite class of equations defined by different temperature dependence:
\begin{equation}\label{eq:ft_linear_response_p}
\begin{gathered}
    \Big[ \widetilde{\mathcal{S}}_p - \omega M \Big] \delta \widetilde{R}_p(\omega) =  - T^p \mathcal{F}(\omega)
    \\
    \widetilde{\mathcal{S}}_p \equiv T^p \mathcal{H} T^{1-p} + \mathcal{E}
    ,\quad
  \delta \widetilde{R}_p \equiv T^p \delta R
  \,.
\end{gathered}
\end{equation}
Equation~\eqref{eq:ft_linear_response} corresponds to $p=1$. The only member of this class of equations that contains a Hermitian $\widetilde{\mathcal{S}}_p$ is the one where $p=1/2$, which is exactly the temperature-symmetric set of equations discussed in Ref.~\cite{Sommermann1983}.

We note that if $T$ is positive semidefinite (e.g., due to degeneracies) this gives rise to redundant degrees of freedom for which the linear response equations read ${0=0}$. In these situations, we must work in the reduced space of non-trivial degrees of freedom for the above arguments to hold. Otherwise a power of $T$ cannot formally be factored out of Eq.~\eqref{eq:ft_linear_response}.

Before moving on to discuss the FT-QRPA equations, we wish to mention an additional way to view Eq.~\eqref{eq:ft_linear_response}. Aside from the $p=1/2$ formulation discussed above, there is only one other possible formulation of the finite-temperature linear response equations that is based on a Hermitian matrix. It comes from grouping the temperature dependence into the metric to re-write the equations as,
\begin{equation}\label{eq:ft_linear_response_m}
\begin{gathered}
\Big[ \widetilde{\mathcal{S}}_{M} - \omega \widetilde{M} \Big] {\delta R}(\omega) =  - T \mathcal{F}(\omega)
\\
\widetilde{\mathcal{S}}_{M} \equiv T \mathcal{H} T + \mathcal{E}T
,\quad
\widetilde{M} \equiv MT
\,.
\end{gathered}
\end{equation}
The temperature-dependent metric is also used, for example, in the thermal RPA theory developed in Ref.~\cite{TANABE1986129}. As we will show, the properties of the FT-QRPA matrix arising from this formulation are more natural to compare to those of the zero-temperature equations. 

We emphasize that, while Eq.~\eqref{eq:ft_linear_response} is always true (it is simply the linear response of the FT-HFB theory), the other formulations based on $\widetilde{\mathcal{S}}_M$ or $\widetilde{\mathcal{S}}_{p\neq 1}$ require careful treatment of the matrix $T$. In the latter formulation, it must be constructed in a basis ordered such that $T$ is positive definite so that $T^{p}$ is well-defined for non-integer $p$. Moreover, in both formulations we must be careful to work only with non-trivial degrees of freedom, which may be a sub-block of the full two-quasiparticle space. The original formulation, Eq.~\eqref{eq:ft_linear_response}, is more useful in approaches that solve the linear response equations directly, such as those based on the finite amplitude method~\cite{Nakatsukasa2007,Avogadro2011,Mustonen2014}. Formulations based on the Hermitian matrices $\widetilde{\mathcal{S}}_M$ or $\widetilde{\mathcal{S}}_{p= 1/2}$ are better used in matrix methods which solve the FT-QRPA eigenvalue problem. Thus, the connection between all formulations is important to illustrate. In the following sections we focus on the original formulation, Eq.~\eqref{eq:ft_linear_response}, and then discuss how it relates to the $p$-dependent (Eq.~\eqref{eq:ft_linear_response_p}) and temperature-dependent metric (Eq.~\eqref{eq:ft_linear_response_m}) formulations.

\subsection{Finite-temperature QRPA}

The free response of Eq.~\eqref{eq:ft_linear_response} results in the FT-QRPA equations, a non-Hermitian eigenvalue problem which reads,
\begin{equation}\label{eq:ftqrpa}
    \big( M \widetilde{\mathcal{S}}\, \big) \,  \delta \widetilde{R}_k = \Omega_k\, \delta \widetilde{R}_k
    \,.
\end{equation}
Again, we see that by taking the appropriate basis order and non-trivial sub-space we can eliminate a factor of $T^{1-p}$ on the left from both sides of this equation to get a family of eigenvalue problems for the matrix $M \widetilde{\mathcal{S}}_p$. These matrices share the same eigenvalues as in Eq.~\eqref{eq:ftqrpa}, but have eigenvectors $\delta \widetilde{R}^k_p = T^p \delta R^k$ that differ in the temperature dependence. Similarly, under the same assumptions we may multiply both sides of Eq.~\eqref{eq:ftqrpa} by $T^{-1}$ to show that $\widetilde{M}^{-1} \widetilde{\mathcal{S}}_M$ also shares the same eigenvalues and has temperature-independent eigenvectors $\delta R_k$.

Having different eigenvectors for the different formulations raises the question as to whether these eigenvalue problems are strictly equivalent. To address this, we examine several properties of the $p$-dependent class of equations in Appendix~\ref{sec:appa}. There we show that all physical quantities are independent of $p$, and therefore the formulations are indeed equivalent. For example, we show that the scalar product of left and right eigenvectors is independent of $p$. Furthermore, if the conditions are met such that the eigenvectors are orthogonal, this scalar product defines the normalization condition
\begin{equation}\label{eq:orthonormality}
    (\delta R^k)^\dagger M T \, \delta R^l = \delta_{kl}
    \,,
\end{equation}
which is also independent of $p$. The same results follow immediately for solutions of the temperature-dependent metric formulation, where the eigenvectors are temperature-independent and the metric becomes $\widetilde{M}=MT$.

It is also instructive to discuss here the equivalence of the stability condition for the different formulations. In Appendix~\ref{sec:appa} we prove that the stability for all formulations depends only on the Hermitian matrix $\widetilde{\mathcal{S}}_M$, and can be stated concisely as,
\begin{equation}\label{eq:ft_stability}
    \widetilde{\mathcal{S}}_M > 0
    \,.
\end{equation}
If $\widetilde{\mathcal{S}}_M$ is positive-definite, it follows that the eigenvalues of $\widetilde{M}^{-1} \widetilde{\mathcal{S}}_M$ and $M \widetilde{\mathcal{S}}_p$ are real, and the sign of the eigenvalue matches the sign of the norm of the corresponding eigenvector. An identical statement can be made about the zero-temperature QRPA~\cite{Ripka1986}, whose stability matrix is $\mathcal{S}$ and eigenvalue problem is for $M \mathcal{S}$.

With these considerations, we argue that the formulation that uses the temperature-dependent metric is the most natural version of the FT-QRPA. The other formulations require some careful treatment of the temperature dependence in the eigenvectors, which can differ between left and right eigenvalue problems for non-Hermitian $\tsp$. On the other hand, $\widetilde{\mathcal{S}}_M$ is both Hermitian and coincides with the stability matrix, and $\widetilde{M}^{-1} \widetilde{\mathcal{S}}_M$ has temperature-independent eigenvectors. Therefore the properties of FT-QRPA eigenvalue problem in this formulation directly mirror those of the zero-temperature problem, so long as we work in the expanded $4\times4$ supermatrix space and use the temperature-dependent metric $\widetilde{M}$.

\subsection{FT-QRPA strength function}

In this section we discuss the FT-QRPA transition strength function.
With the spectral decomposition of the FT-QRPA matrix, we can formally invert $[\widetilde{\mathcal{S}} - \omega M]$ in Eq.~\eqref{eq:ft_linear_response} to express the response function (i.e., the Greens function), $\widetilde{G}(\omega)$, in terms of FT-QRPA eigen-solutions. The response function is defined as the connection between the external field and the density response, $\delta \widetilde{R}(\omega) = \widetilde{G}(\omega) \mathcal{F}$. In the case of Eq.~\eqref{eq:ft_linear_response}, it is,
\begin{equation}\label{eq:ft_greens}
    \widetilde{G}(\omega) = - T {\mathcal{X}} (\widetilde{\mathcal{O}} - \omega)^{-1}  M' {\mathcal{X}}^{\dagger} T
    \,,
\end{equation}
where $T{\mathcal{X}}$ is a matrix whose columns are right eigenvectors of $M\widetilde{\mathcal{S}}$, and $\mathcal{O}=\text{diag}(\Omega_{k-}, \Omega_{k+}, -\Omega_{k+}, -\Omega_{k-})$ is a matrix of the corresponding eigenvalues~\cite{Sommermann1983}. To write Eq.~\eqref{eq:ft_greens} we used the orthogonality relation (cf.~Appendix~\ref{sec:appa} Property~\ref{prop:orthogonal}) to express the inverse of $T{\mathcal{X}}$ in terms of its adjoint, $(T{\mathcal{X}})^{-1} = M' {\mathcal{X}}^\dagger M$, where $M'$ has the same form as $M$ but has the dimension of the eigenvalue problem.

The eigenvalues $\Omega_{k+}$ are equal to $E_\alpha + E_\beta$ in the FT-HFB limit and represent the usual transitions from the ground state to excited states. $\Omega_{k-}$ are new at finite temperature. They equal $E_\alpha - E_\beta$ in the FT-HFB limit, and come from excitations among thermally populated excited states. If the stability condition in Eq.~\eqref{eq:ft_stability} is met, then $\Omega_{k\pm} > 0$ and the $\Omega_{k-}$ fall in-between the $\Omega_{k+}$.

Now, the transition strength function contains squared transition amplitudes. The zero-temperature strength is typically expressed as $\lvert \bra{k} \hat{F} \ket{0} \rvert^2$, where $\ket{k}$ is a QRPA phonon, $\ket{0}$ is the QRPA ground state, and $\hat{F}$ is an external field. For an analogous expression at finite temperature, we examine the FT-QRPA equations of motion~\cite{Sommermann1983},
\begin{equation}
    \left\langle \left[\delta \Gamma^k, \left[H,\Gamma^{k \dagger}\right]\right] \right\rangle = \Omega_k \left\langle \left[ \delta \Gamma^k, \Gamma^{k \dagger} \right] \right\rangle
\end{equation}
where 
\begin{equation}\label{eq:ft_phonon_operator}
    \Gamma^{k \dagger} = \sum\limits_{\beta < \alpha} P^k_{\alpha \beta} a^\dagger_\alpha a_\beta + X^k_{\alpha \beta} a^\dagger_\alpha a^\dagger_\beta - Y^{k *}_{\alpha \beta} a_\beta a_\alpha - Q^{k *}_{\alpha \beta} a^\dagger_\beta a_\alpha
\end{equation}
is an operator that creates a phonon in an ensemble where quasiparticle $\alpha$ has occupation $f_\alpha$. The ensemble average $\langle \rangle$ means to trace with the FT-HFB statistical density operator,
\begin{equation}\label{eq:statistical_density}
    \langle \mathcal{O} \rangle = \text{Tr}[\mathcal{O} D]
\end{equation}
and $D$ takes the well known form~\cite{Goodman1981},
\begin{equation}\label{eq:indep_part_density}
    D = \frac{e^{-\beta H_{\text{HFB}}}}{\text{Tr}[e^{-\beta H_{\text{HFB}}}]}
    =
    \frac{\prod_\alpha \left[ f_\alpha a^\dagger_\alpha a_\alpha 
    + (1- f_\alpha) a_\alpha a^\dagger_\alpha \right]}{\prod_\alpha \left[ 1 + e^{-\beta E_\alpha} \right]}
    \,.
\end{equation}
From the equations of motion, the FT-QRPA amplitudes can be written~\cite{Sommermann1983}
\begin{equation}\label{eq:FTamplitudes}
\begin{aligned}
    \left \langle \left[ a^\dagger_\beta a_\alpha ,\Gamma^{k\dagger} \right] \right \rangle
    &= (f_\beta - f_\alpha) P^k_{\alpha \beta} \\
    \left \langle \left[ a_\beta a_\alpha ,\Gamma^{k\dagger} \right] \right \rangle
    &= (1 - f_\alpha - f_\beta) X^k_{\alpha \beta} \\
    \left \langle \left[ a^\dagger_\alpha a^\dagger_\beta, \Gamma^{k\dagger} \right] \right \rangle
    &= (1 - f_\alpha - f_\beta) Y^k_{\alpha \beta} \\
    \left \langle \left[ a^\dagger_\alpha a_\beta ,\Gamma^{k\dagger} \right] \right \rangle
    &= (f_\beta - f_\alpha) Q^k_{\alpha \beta}
    \,.
\end{aligned}
\end{equation}
and we can express the finite-temperature analogue of the zero-temperature transition matrix element $\bra{0} \hat{F} \ket{k}$ as,
\begin{equation}\label{eq:F_ensemble_avg}
\begin{aligned}
    &\left\langle \left[ \hat{F}, \hat{\Gamma}^{k \dagger} \right] \right\rangle 
    =
    \mathcal{F}^\dagger \delta \widetilde{R}^k
    \\&=
    \sum\limits_{\beta < \alpha} 
        F^{{11}}_{\alpha\beta} (f_\beta-f_\alpha)P^{k}_{\alpha\beta}  
      + F^{02}_{\alpha\beta} (1-f_\beta-f_\alpha)X^{k}_{\alpha\beta}   
        \\&\qquad
      + F^{20}_{\alpha\beta} (1-f_\beta-f_\alpha)Y^{k}_{\alpha\beta}
      + F^{\bar{11}}_{\alpha\beta} (f_\beta-f_\alpha)Q^{k}_{\alpha\beta} 
        \,.
\end{aligned}
\end{equation}
As discussed in Appendix \ref{sec:appa}, this physical result can be obtained from any formulation by tracing the eigenvector with the appropriate factor of $T$.

From here it is evident that a function of squared transition amplitudes must be quadratic in $T$ and $\mathcal{F}$. For Eqs.~\eqref{eq:ft_linear_response} and~\eqref{eq:ft_greens} we need only trace the density response with the external field, which gives,
\begin{equation}\label{eq:ft_fam_strength}
\begin{aligned}
    \widetilde{S}_F(\omega) 
    &= 
    \mathcal{F}^\dagger \widetilde{G}(\omega) \mathcal{F}
    =    
    \mathcal{F}^\dagger \delta \widetilde{R}(\omega)
    \\&
    = 
    \sum\limits_{k \pm >0} \left[ 
    \frac{
    \big\lvert 
        \big\langle \big[ 
            \Gamma^{k}, \hat{F}
        \big] \big\rangle
    \big\rvert^2}
         {\omega - \Omega_k} -
    \frac{
    \big\lvert
        \big\langle \big[ 
            \Gamma^{k}, \hat{F}^\dagger
        \big] \big\rangle    
    \big\rvert^2}
        {\omega + \Omega_k} \right]
\end{aligned}
\end{equation}
where the sum is over FT-QRPA modes with positive eigenvalues. 

As for the $p$-dependent formulations, the Greens function is not unique. It becomes ${\widetilde{G}_p(\omega) = - T^p {\mathcal{X}} (\widetilde{\mathcal{O}} - \omega)^{-1}  M' {\mathcal{X}}^{\dagger} T}
$, and the strength function is thus obtained as $\widetilde{S}_F(\omega) = (T^{1-p} \mathcal{F})^\dagger \delta \widetilde{R}_p(\omega)$. Similarly, for the temperature-dependent metric formulation, the Greens function takes the form ${\widetilde{G}_M(\omega) = - {\mathcal{X}} (\widetilde{\mathcal{O}} - \omega)^{-1}  M' {\mathcal{X}}^{\dagger} T}
$, and the strength function is $\widetilde{S}_F(\omega) = (T \mathcal{F})^\dagger \delta \widetilde{R}_M(\omega)$. In all cases, we are able to arrive at the same strength function, Eq.~\eqref{eq:ft_fam_strength}.

\section{Transitions in the FT-QRPA}\label{sec:transitions}

\subsection{Exact transition strength}\label{ssec:exact_trans_str}
To elucidate the meaning of some of the quantities discussed in the previous section, we briefly review some expressions from the exact linear response of a finite-temperature ensemble.
The exact finite-temperature response function is~\cite{Vautherin1984a,Chomaz1990}
% \begin{widetext}
\begin{equation}
\begin{aligned}
    \widetilde{G}^e_{\mu\mu',\nu\nu'} = \frac{1}{Z} \sum\limits_{i,f} &e^{-\beta \omega_i} \Bigg[ \frac{\bra{i} a^\dagger_{\mu'} a_\mu \ket{f}\bra{f} a^\dagger_{\nu} a_{\nu'} \ket{i}} {\omega - (\omega_f- \omega_i)} 
    \\&
    -
    \frac{\bra{i} a^\dagger_{\nu} a_{\nu'} \ket{f}\bra{f} a^\dagger_{\mu'} a_{\mu} \ket{i}} {\omega + (\omega_f- \omega_i)}
    \Bigg]
\end{aligned}
\end{equation}
% \end{widetext}
where $Z$ is the partition function, $\ket{i}$ and $\ket{f}$ are exact eigenstates of the Hamiltonian, and $\omega_i$ and $\omega_f$ are their eigenvalues. The exact finite-temperature strength function is then,
\begin{equation}\label{eq:s_exact_ft}
\begin{aligned}
    \widetilde{S}^e(\omega) &= \sum\limits_{\mu'<\mu,\nu'<\nu} \mathcal{F}^\dagger_{\mu\mu'} \widetilde{G}^e_{\mu\mu',\nu\nu'} \mathcal{F}_{\nu\nu'}
    \\ &=
    \frac{1}{Z}
        \sum\limits_{i,f} e^{-\beta \omega_i} \left[ 
    \frac{
    \big\lvert 
        \bra{f} \hat{F} \ket{i}
    \big\rvert^2}
         {\omega - (\omega_f - \omega_i)} -
    \frac{
    \big\lvert
        \bra{f} \hat{F}^\dagger \ket{i}
    \big\rvert^2}
        {\omega + (\omega_f - \omega_i)} \right]
        \,.
\end{aligned}
\end{equation}
On the real axis, the imaginary part of this function is related to the distribution,
% \begin{widetext}
\begin{equation}\label{eq:exact_ft_transition_str_fct}
\begin{aligned}
    \frac{dB}{d\omega}' 
    &=
    -\frac{1}{\pi} \Im[\widetilde{S}^e(\omega)] 
    \\
    &= 
    \frac{1}{Z}
    \sum\limits_{i,f} e^{-\beta \omega_i} \Big[ \lvert \bra{f}  \hat{F} \ket{i} \rvert^2 \delta (\omega - (\omega_f - \omega_i))
    \\&\qquad\qquad
    -\lvert \bra{f} \hat{F}^\dagger \ket{i} \rvert^2 \delta (\omega + (\omega_f - \omega_i)) \Big]
    \,,
\end{aligned}
\end{equation}
% \end{widetext}
where all terms in Eq.~\eqref{eq:exact_ft_transition_str_fct} are defined at both positive and negative $\omega$. The prime emphasizes that Eq.~\eqref{eq:exact_ft_transition_str_fct} differs from the physical strength distribution $dB/d\omega$. At zero temperature, we need only to consider positive $\omega$, and the physical transition strength distribution is directly related to the imaginary part of the strength function at these energies. However, at finite temperature there is the possibility that a thermally populated excited state will transition to a state of lower energy. This is evident by the double sum over both $i$ and $f$, which implies that both excitations --- transitions from $\ket{i}$ to $\ket{f}$ --- as well de-excitations --- transitions from $\ket{f}$ to $\ket{i}$ --- are included. The de-excitation transitions have negative energies.

To make the distinction between excitations and de-excitations more obvious, we rewrite Eq.~\eqref{eq:exact_ft_transition_str_fct} with a sum over only unique pairs of states. We then have terms for excitations induced by the external field $\hat{F}$ (denoted with a superscript $+$),
\begin{equation}\label{eq:separate_strengths_exc}
\begin{aligned}
    \frac{dB}{d\omega}^+(\hat{F}) &= 
    \frac{1}{Z} \sum\limits_{i<f} e^{-\beta \omega_i} \lvert \bra{f}  \hat{F} \ket{i} \rvert^2 \delta (\omega - (\omega_f - \omega_i))   
    \\
    \frac{dB}{d\omega}^+(\hat{F}^\dagger) &= 
    \frac{1}{Z} \sum\limits_{i<f}e^{-\beta \omega_i} \lvert \bra{f} \hat{F}^\dagger \ket{i} \rvert^2 \delta (\omega + (\omega_f - \omega_i)) 
    \,,
\end{aligned}
\end{equation}
and terms for de-excitations (denoted with a superscript $-$),
\begin{equation}\label{eq:separate_strengths_deexc}
\begin{aligned}
    \frac{dB}{d\omega}^-(\hat{F}) &= 
    \frac{1}{Z} \sum\limits_{i<f} e^{-\beta \omega_f} \lvert \bra{i}  \hat{F}\ket{f} \rvert^2 \delta (\omega - (\omega_i - \omega_f))  
    \\
    \frac{dB}{d\omega}^-(\hat{F}^\dagger) &= 
    \frac{1}{Z} \sum\limits_{i<f}e^{-\beta \omega_f} \lvert \bra{i} \hat{F}^\dagger \ket{f} \rvert^2 \delta (\omega + (\omega_i - \omega_f)) 
    \,.
\end{aligned}
\end{equation}
The full Eq.~\eqref{eq:exact_ft_transition_str_fct} can then be re-written as,
\begin{equation}\label{eq:strength_de-ex}
    \frac{dB}{d\omega}' = \left( \frac{dB}{d\omega}^+(\hat{F}) +  \frac{dB}{d\omega}^-(\hat{F}) \right) 
    - \left(\frac{dB}{d\omega}^+(\hat{F}^\dagger) +  \frac{dB}{d\omega}^-(\hat{F}^\dagger) \right)
    \,.
\end{equation}

Let us now consider a single energy $\omega_{fi} = \omega_f - \omega_i$. The transition strength evaluated at this energy reads,
\begin{equation}\label{eq:separate_strengths}
\begin{aligned}
    \frac{dB}{d\omega}' \Bigg\rvert_{\omega = \omega_{fi}} =
    \left(
    \frac{dB}{d\omega}^+(\hat{F})  
    - \frac{dB}{d\omega}^-(\hat{F}^\dagger)\right) \Bigg\rvert_{\omega = \omega_{fi}}
    \,.
\end{aligned}
\end{equation}
This demonstrates that Eq.~\eqref{eq:exact_ft_transition_str_fct} on its own is not exactly equal to the \textit{physical} strength distribution. At a given energy, the de-excitation strength for the reverse process governed by $\hat{F}^\dagger$ interferes with the excitation strength for the forward process, $\hat{F}$, and vice versa.
We can, however, eliminate the reverse process contributions very easily. Since $\lvert \bra{f} \hat{F} \ket{i} \rvert^2 = \lvert \bra{i} \hat{F}^\dagger \ket{f} \rvert^2$, and through the delta function $\omega = \omega_{fi}$, we can show that
\begin{equation}\label{eq:relate_F_and_R_procs}
\begin{aligned}
    e^{-\beta \omega}
    \frac{dB}{d\omega}^+(\hat{F}) \Bigg\rvert_{\omega = \omega_{fi}} 
    &=
    \frac{dB}{d\omega}^-(\hat{F}^\dagger)  \Bigg\rvert_{\omega = \omega_{fi}}
    \,,
\end{aligned}
\end{equation}
where $e^{-\beta \omega} = e^{-\beta \omega_f}/e^{-\beta \omega_i}$ is simply the ratio of the ensemble weights for the states $\ket{i}$ and $\ket{f}$ involved in the transition.
Such a relation is often assumed on the basis of detailed balance~\cite{Egido1993,Sommermann1983,PhysRevC.92.045804}, but for Eq.~\eqref{eq:exact_ft_transition_str_fct} it holds even without this assumption. The expression here and from detailed balance have identical forms because, in the case of the grand canonical ensemble, the ensemble weights are Boltzmann factors.

Thus, for a given $\omega$, the second term in Eq.~\eqref{eq:s_exact_ft} contributes the same as the first term, up to the factor $e^{-\beta \omega}$. We can therefore express the strength function as,
\begin{equation}\label{eq:s_exact_detailed_bal}
    \widetilde{S}^e(\omega) =
    \frac{1}{Z}
        \sum\limits_{i,f} e^{-\beta \omega_i} 
    \frac{
    \big\lvert 
        \bra{f} \hat{F} \ket{i}
    \big\rvert^2}
         {\omega - (\omega_f - \omega_i)} 
     \left[ 1 - e^{-\beta\omega} \right]
        \,.
\end{equation}
It is now clear that the physical strength distribution for the forward process (which still contains both excitations and de-excitations), is obtained from $\widetilde{S}^e(\omega)$ through
\begin{equation}\label{eq:ft_prefactor}
\begin{aligned}
    \frac{dB}{d\omega} &= -\frac{1}{\pi} \Im\left[\frac{\widetilde{S}^e(\omega)}{1 - e^{-\beta \omega}}\right]
    \\
    &=
    \frac{1}{Z} \sum\limits_{i,f} e^{-\beta \omega_i} \big\lvert \bra{f} \hat{F} \ket{i} \big\rvert^2 \delta(\omega - (\omega_f - \omega_i))
    \,.
\end{aligned}
\end{equation}

To summarize, we have shown that some of the physical strength comes from de-excitations, which are located at $\omega < 0$, in addition to excitations at $\omega > 0$. The imaginary part of the response function is related to a distribution that contains contributions from the forward and reverse processes. At zero temperature, these contributions are well separated, but at finite temperature they interfere. However, the interfering contributions are related by their ensemble weights, which can be used to eliminate the unwanted strength. This is the source of the thermal prefactor, $(1 - e^{-\beta \omega})^{-1}$, which must be included to obtain the physical strength distribution, as in Eq.~\eqref{eq:ft_prefactor}.

\subsection{Comparison to FT-QRPA}\label{ssec:comparison_ftqrpa}

We now wish to leverage our understanding of the exact transition strength distribution to identify where the various contributions appear in the FT-QRPA. From Eq.~\eqref{eq:ft_fam_strength} we see that the exact distribution in Eq.~\eqref{eq:exact_ft_transition_str_fct} is approximated in the FT-QRPA by 
\begin{equation}\label{eq:qrpa_ft_strength}
    \frac{dB}{d\omega}' \approx
    \sum_{k\pm>0}
    \big\lvert 
        \big\langle \big[ 
            \Gamma^{k}, \hat{F}
        \big] \big\rangle
    \big\rvert^2 \delta(\omega - \Omega_k )
    -
    \big\lvert 
        \big\langle \big[ 
            \Gamma^{k}, \hat{F}^\dagger
        \big] \big\rangle
    \big\rvert^2 \delta(\omega + \Omega_k )
 \,.
\end{equation}
Using Eq.~\eqref{eq:strength_de-ex} to relate FT-QRPA expressions to the exact expression at the same energies, we have
\begin{equation}\label{eq:QRPA_str_numerators}
\begin{aligned}
    \frac{dB}{d\omega}^+(\hat{F}) &- \frac{dB}{d\omega}^-(\hat{F}^\dagger)
    \\
    &\approx
    \sum\limits_{k\pm>0}
    \big\lvert 
        \big\langle \big[ 
            \Gamma^{k}, \hat{F}\phantom{^{\dagger}}
        \big] \big\rangle
    \big\rvert^2 
    \delta(\omega-\Omega_k),
    \quad
    \omega > 0
    \\
    \frac{dB}{d\omega}^+(\hat{F}^\dagger) &- \frac{dB}{d\omega}^-(\hat{F})
    \\
    &\approx
    \sum\limits_{k\pm>0}
    \big\lvert 
        \big\langle \big[ 
            \Gamma^{n}, \hat{F}^\dagger
        \big] \big\rangle
    \big\rvert^2
    \delta(\omega+\Omega_k),
    \quad
    \omega < 0
    \,.
\end{aligned}
\end{equation}
Coupled with Eq.~\eqref{eq:relate_F_and_R_procs}, we can now make the following claims about what the FT-QRPA quantities approximate:
\begin{widetext}
\begin{equation}\label{eq:FTQRPA_strength_meaning}
\begin{aligned}
        \frac{
         \big\lvert \big\langle \big[ 
            \Gamma^{k}, \hat{F}^{\phantom{\dagger}}
        \big] \big\rangle \big\rvert^2
        }{1 - e^{-\beta \omega}}
        &\approx
        \frac{1}{Z} \sum\limits_{i<f} e^{-\beta \omega_i}
        \lvert \bra{f} \hat{F} \ket{i} \rvert^2 \quad \forall \quad (\omega_f-\omega_i) \approx \Omega_k > 0
    \\
        -\frac{
        \big\lvert \big\langle \big[ 
            \Gamma^{k}, \hat{F}^\dagger
        \big] \big\rangle \big\rvert^2
        }{1 - e^{-\beta \omega}}
        &\approx
        \frac{1}{Z} \sum\limits_{i<f} e^{-\beta \omega_f}
        \lvert \bra{i} \hat{F} \ket{f} \rvert^2 \quad \forall \quad (\omega_i-\omega_f) \approx - \Omega_k < 0
        \,.
\end{aligned}
\end{equation}
\end{widetext}
Clearly, in the FT-QRPA we still need to include strength due to de-excitations at $\omega < 0$. We also need the thermal prefactor to eliminate interference due to the reverse process. We should therefore continue to use Eq.~\eqref{eq:ft_prefactor} to compute the physical strength distribution from the FT-QRPA strength function.

The relations in Eq.~\eqref{eq:FTQRPA_strength_meaning} are somewhat unusual because some of the physical strength comes from \textit{both} terms in the strength function. In contrast, at zero temperature only the term at $\omega > 0$ is used, and the one at $\omega < 0$ is considered unphysical and neglected. Further adding to the confusion, at finite temperature the de-excitations are located at $\omega < 0$, but all stable FT-QRPA eigenvalues are positive. To clarify these issues, we examine the energy-weighted sum rule.

The exact sum rule is related to the first moment of the physical strength distribution,
\begin{equation}
    \Sigma^1 = \int_{-\infty}^{+\infty} d\omega\ \omega \frac{dB}{d\omega}
    = \frac{1}{Z} \sum\limits_{i,f} e^{-\beta \omega_i} \lvert \bra{f} \hat{F} \ket{i} \rvert^2 (\omega_f - \omega_i)
    \,.
\end{equation}
This can be expressed as a double commutator, which reads,
\begin{equation}
    \Sigma^1 = \frac{1}{2} \left\langle \left[\hat{F}^\dagger, \left[H, \hat{F}\right]\right] \right\rangle
    \,.
\end{equation}
In the FT-QRPA we can evaluate the double commutator using a generalization of Thouless' theorem~\cite{Thouless1961, Sommermann1983, Ring2004},
\begin{equation}
\begin{aligned}
    \left\langle \left[\hat{F}^\dagger, \left[H, \hat{F}\right]\right] \right\rangle 
    &=
     \sum\limits_{k\pm>0}
    \Omega_{k} \left(
        \left\lvert\langle [\Gamma^k, \hat{F}] \rangle\right\rvert^2
       +\left\lvert\langle [\Gamma^k, \hat{F}^\dagger] \rangle\right\rvert^2\right)
       \,.
\end{aligned}
\end{equation}
% The zero-temperature analogue of this relation is straightforward to interpret, but at finite temperature the ensemble averages complicate the matter. 
Using the relations in Eq.~\eqref{eq:QRPA_str_numerators}, with a little algebra we can show that the FT-QRPA sum rule approximates the exact relation,
\begin{equation}
\begin{aligned}
    &\sum\limits_{k\pm>0} \Omega_{k} \left(
        \left\lvert\langle [\Gamma^k, \hat{F}] \rangle\right\rvert^2
       +\left\lvert\langle [\Gamma^k, \hat{F}^\dagger] \rangle\right\rvert^2\right)
    \\
    &\approx
    \Big(\frac{dB}{d\omega}^+(\hat{F})(\omega_{f}-\omega_{i}) +\frac{dB}{d\omega}^-(\hat{F})(\omega_{i}-\omega_{f})\Big)
    % \\&\quad
    +
    h.c.
    % \Big(\frac{dB(\hat{F}^\dagger_+)}{d\omega}(\omega_{f}-\omega_{i}) + \frac{dB(\hat{F}^\dagger_-)}{d\omega}(\omega_{i}-\omega_{f}) \Big)
    % \,.
\end{aligned}
\end{equation}
Note that the $\omega_f -\omega_i$ terms fall under the sums in Eqs.~\eqref{eq:separate_strengths_exc} and~\eqref{eq:separate_strengths_deexc}. Through this exercise we can see how even though we only sum over $\Omega_{k\pm}>0$, because the ensemble averages contain an interference between forward and reverse processes, we still get the correct contributions from de-excitations at $\omega < 0$. They originate from the term proportional to $\hat{F}^\dagger$ in the FT-QRPA strength function, Eq.~\eqref{eq:ft_fam_strength}.

\section{Gamow-Teller strength and sum rule}\label{sec:bgt}

In this section we apply the formalism outlined in this work to the calculation of the charge-exchange Gamow-Teller (GT) strength functions. The GT transitions correspond to coupling the total angular momentum and parity to $J^\pi = 1^+$ with the total isospin $T = 0$ and spin $S = 1$. The external field operator takes the well known form $\boldsymbol{\sigma} \tau_\pm$ representing GT${}^\pm$ excitations, where $\boldsymbol{\sigma}$ is the Pauli spin matrix and $\tau_\pm$ the isospin raising (lowering) operator. The GT${}^+$ transition corresponds to the change of isospin projecton $\Delta T_z = +1$ and describes transitions between proton to neutron states $p \to n$, while the opposite is true for GT${}^-$. It is important to note that GT transitions connect nuclei with different charges and therefore different ground states. We demonstrate that the Ikeda sum rule \cite{IKEDA1963271} is satisfied within the present formalism and perform calculations of the GT strength in ${}^{58}$Ni employing both the relativistic matrix FT-PNQRPA and non-relativistic finite-temperature charge-changing finite amplitude method (FT-PNFAM).

\subsection{Ikeda sum rule with FT-QRPA}

The zeroth moment of the physical strength distribution is given by
\begin{equation}
\Sigma^0 = \int \limits_{-\infty}^\infty d \omega \frac{d B}{d \omega} = \frac{1}{Z} \sum \limits_{if} e^{-\beta \omega_i} |\langle f | \hat{F}| i \rangle |^2.
\end{equation}
In order to derive the Ikeda sum rule, the external field operator assumes the GT form $\hat{F} = \boldsymbol{\sigma} \tau_-$.
% for which $\hat{F} \neq \hat{F}^\dag$, where $\boldsymbol{\sigma}$ is the Pauli spin matrix and $\tau_-$ isospin lowering operator.
The Ikeda sum rule is then defined by the difference~\cite{IKEDA1963271}
\begin{align}
\begin{split}
&\Sigma^0(\hat{F}) - \Sigma^0(\hat{F}^\dag) 
%= \frac{1}{Z} \sum \limits_{if} e^{-\beta \omega_i} \left[ |\langle f | \hat{F}| i \rangle |^2 - |\langle f | \hat{F}^\dag| i \rangle |^2 \right] \\
%&
%= \sum \limits_f \langle f | \hat{F} \hat{D} \hat{F}^\dag | f \rangle - \langle f | \hat{F}^\dag \hat{D} \hat{F} | f \rangle 
= \langle [ \hat{F}^\dag, \hat{F} ] \rangle,
\end{split}
\end{align}
where we have used the definition of the finite-temperature density operator $\hat{D} = \frac{1}{Z} \sum \limits_i e^{-\beta \omega_i} | i \rangle \langle i |$ and definition of the thermal average (cf.~Eq.~\eqref{eq:statistical_density}). At this point we approximate the thermal average using the statistical density operator of independent quasiparticles in Eq.~\eqref{eq:indep_part_density}. The external field operator in the proton-neutron quasiparticle basis takes the form
\begin{equation}\label{eq:ftqrpa_field_op}
\hat{F} = \sum \limits_{\pi \nu} F_{\pi \nu}^{11} a_\pi^\dag a_\nu + F^{20}_{\pi \nu} a_\pi^\dag a_\nu^\dag + F^{02}_{\pi \nu} a_\pi a_\nu + F^{\bar{11}}_{\pi \nu} a_\pi a_\nu^\dag,
\end{equation}
where $\pi$ ($\nu$) denote proton and neutron quasiparticles, respectively. We can evaluate the ensemble averages $\langle \hat{F} \hat{F}^\dag \rangle$ and $\langle \hat{F}^\dag \hat{F} \rangle$ using the expressions $\langle a_\alpha^\dag a_{\beta} \rangle = f_\alpha \delta_{\alpha \beta}$ and ${\langle a_\alpha a_{\beta}^\dag \rangle = (1-f_\alpha) \delta_{\alpha \beta}}$~\cite{Sommermann1983}. Finally, for the commutator we have
\begin{align}\label{eq:thermal_avg_of_comm}
\begin{split}
&\langle [ \hat{F}^\dag, \hat{F} ] \rangle = F^{11}_{\pi \nu} (F^\dag)^{11}_{\pi \nu}(f_\nu - f_\pi) + (1-f_\nu - f_\pi) F^{20}_{\pi \nu} (F^\dag)^{20}_{\pi \nu} \\
&- (1-f_\pi -f_\nu) F^{02}_{\pi \nu} (F^\dag)^{02 }_{\pi \nu} - F^{\bar{11}}_{\pi \nu} (F^\dag)^{\bar{11} }_{\pi \nu} (f_\nu - f_\pi).
\end{split}
\end{align}
The above expression can be evaluated either using the FT-HFB or the FT-HFBCS approximations (neglecting higher-order correlations), which yields the well-known result
\begin{equation}
\langle [ \hat{F}^\dag, \hat{F} ] \rangle = 3(N-Z),
\end{equation}
where $N$ ($Z$) denotes neutron (proton) number. The derivation within the FT-HFB theory is given in Appendix \ref{sec:appb}. 

On the other hand, we can start the derivation from the physical strength distribution as approximated by the FT-QRPA response in Eq.~\eqref{eq:qrpa_ft_strength}. We evaluate the difference between zeroth moments as
\begin{align}\label{eq:ft_qrpa_ikeda_sum_rule}
\begin{split}
&\Sigma^0(\hat{F}) - \Sigma^0(\hat{F}^\dag) = \int \limits_{-\infty}^\infty d \omega \left[ \frac{dB}{d \omega}(\hat{F}) - \frac{d B}{d \omega} (\hat{F}^\dag) \right] \\
%&= \sum \limits_{k\pm>0} \left( \frac{|\langle [\Gamma^k, \hat{F}] \rangle|^2}{1-e^{-\beta \Omega_k}} -  \frac{|\langle [\Gamma^k, \hat{F}^\dag] \rangle|^2}{1-e^{\beta \Omega_k}} \right) \\
%&\qquad\qquad- \left( \frac{|\langle [\Gamma^k, \hat{F}^\dag] \rangle|^2}{1-e^{-\beta \Omega_k}} -  \frac{|\langle [\Gamma^k, \hat{F}] \rangle|^2}{1-e^{\beta \Omega_k}} \right) \\
&= \sum \limits_{k\pm>0} \left( |\langle [ \Gamma^k, \hat{F}] \rangle|^2 - |\langle [ \Gamma^k, \hat{F}^\dag ] \rangle|^2 \right).
\end{split}
\end{align}
It can be shown that the above expression reduces to the thermal average of the commutator as in Eq.~\eqref{eq:thermal_avg_of_comm} (see Appendix \ref{sec:appc} for details). Therefore, we have demonstrated that the Ikeda sum rule is satisfied within the present formalism. In particular, we have shown that when the thermal prefactor is included, as motivated in Sections~\ref{ssec:exact_trans_str} and~\ref{ssec:comparison_ftqrpa}, the Ikeda sum rule is satisfied only when the strength due to de-excitations (at $\omega < 0$) is also included in the sum.

\begin{figure*}
\centering
\includegraphics[width = \linewidth]{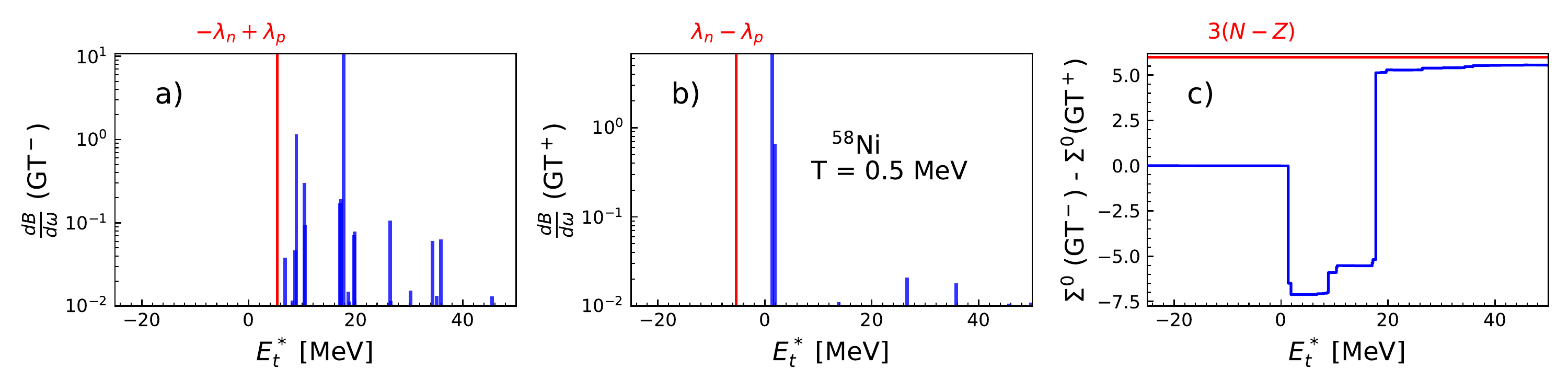}
\\
\includegraphics[width = \linewidth]{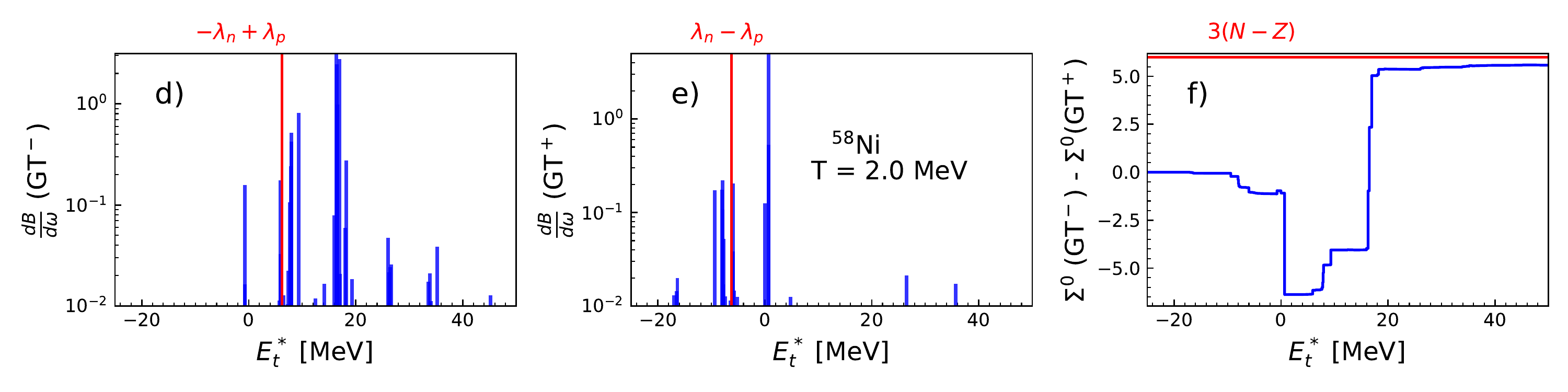}
\caption{The temperature evolution of the Gamow-Teller strength distribution $dB/d\omega$ in ${}^{58}$Ni with respect to the excitation energy $E^*_t$ of the parent nucleus in the (a),(d) GT${}^-$ and (b),(e) GT${}^+$ directions. Results are shown for temperatures (a)--(c) $T = 0.5$~MeV  and (d)--(f) $T = 2.0$~MeV. The threshold between positive and negative energy transitions ($E^*_t = \pm( \lambda_n-\lambda_p)$ for GT${}^\pm$) is plotted with the red vertical line. (c),(f) the difference between zeroth moments $\Sigma^0(\text{GT}^-) - \Sigma^0(\text{GT}^+)$, with the red horizontal line denoting the result of the Ikeda sum rule $3(N-Z)$. Calculations are performed with the FT-HBCS+FT-PNRQRPA model using the D3C${}^*$ interaction.}\label{fig:gt_strength}
\end{figure*}

\subsection{Sum rule with the relativistic FT-PNQRPA}
\label{ssec:rqrpa_sumrule}

The relativistic FT-PNQRPA (FT-PNRQRPA) calculation is based on the relativistic mean field theory where pairing correlations are treated within the finite-temperature Hartree Bardeen-Cooper-Schrieffer (FT-HBCS) theory with a monopole pairing interaction~\cite{PhysRevC.101.044305,PhysRevC.102.065804}. For the mean field part of the Hamiltonian we employ the D3C${}^*$ relativistic EDF~\cite{PhysRevC.75.024304} and assume spherical symmetry. The FT-PNRQRPA eigenvalue problem is derived from Eq.~\eqref{eq:ft_linear_response_m} by expanding the density response $\delta \tilde{\mathcal{R}}$ in the configuration space of the (quasi)proton-(quasi)neutron basis. For a detailed description of the method, check Refs.~\cite{Sommermann1983, Yuksel2017, PhysRevC.101.044305,PhysRevC.104.054318}. From the FT-PNRQRPA eigensolutions we compute the GT transition strengths according to Eq.~\eqref{eq:F_ensemble_avg} and construct the physical strength distribution from them by including the thermal prefactor,
\begin{align}\label{eq:ftpnrqrpa_strength}
\begin{split}
\frac{d B}{d \omega}(\mathrm{GT}^\pm) &= \frac{1}{1 - e^{-\beta \omega}} \sum \limits_{k \pm > 0} \left( |\langle [ \Gamma^k, \vec{\sigma}\tau_\pm ] \rangle|^2 \delta (\omega- \Omega_k)  \right. \\
&-  \left.|\langle [ \Gamma^k, \vec{\sigma}\tau_\mp ] \rangle|^2 \delta (\omega+ \Omega_k)   \right).
\end{split}
\end{align}
% Thermal averages of $\langle [ \Gamma^k, \hat{F} ] \rangle$ and $\langle [ \Gamma^k, \hat{F}^\dag ] \rangle$ are defined in Eqs.~\eqref{eq:ensamble_averages_1_appc} and~\eqref{eq:ensamble_averages_2_appc}, respectively. 

To study the temperature evolution of the GT strength function and for the numerical check of the Ikeda sum rule, we select ${}^{58}$Ni and perform calculations of the GT${}^\pm$ strength functions at temperatures  $T = 0.5$ and 2~MeV. For this calculation, the same numerical cut-offs are used as in Ref.~\cite{PhysRevC.105.055801}. Namely, the nuclear ground state at finite temperature is obtained by solving the FT-HBCS equations in the spherical harmonic oscillator basis with $N_{\mathrm{osc}} = 20$ oscillator shells. At the FT-PNRQRPA level, the maximal two-quasiparticle excitation energy is cut-off at 100~MeV, i.e., ${E_\pi + E_\nu < 100~\mathrm{MeV}}$, and threshold for the product of FT-HBCS occupation amplitudes is set to $u_{(\pi,\nu)} v_{(\nu,\pi)} > 0.01$. The monopole pairing strength for neutrons in ${}^{58}$Ni is $G_n = 24.3$~MeV/A, adjusted to reproduce the pairing gap calculated using the five-point formula~\cite{bender2000pairing}. Note that due to the shell closure at $Z = 20$, no pairing occurs for proton states. Results are displayed in Fig.~\ref{fig:gt_strength}\hyperref[fig:gt_strength]{(a)--(c)}, for $T = 0.5$ MeV and Fig.~\ref{fig:gt_strength}\hyperref[fig:gt_strength]{(d)--(f)} for $T = 2.0$ MeV. We plot the GT$^\pm$ strength distributions as functions of the excitation energy with respect to the parent (i.e., the target) nucleus, $E^*_t = \omega \pm (\lambda_n - \lambda_p)$
\footnote{Strength functions are also viewed as a function of the daughter excitation energy $E^*_d$. The energies are related by $E_t^* = E_d^* + \mathrm{BE}_d - \mathrm{BE}_t$, where BE$<0$ are binding energies of the daughter and target, respectively.}.
The GT${}^-$ strength appears in Fig.~\ref{fig:gt_strength}\hyperref[fig:gt_strength]{(a),(d)}, the GT${}^+$ strength in Fig.~\ref{fig:gt_strength}\hyperref[fig:gt_strength]{(b),(e)}, and the difference between zeroth moments $\Sigma^0(\text{GT}^-) - \Sigma^0(\text{GT}^+)$ in Fig.~\ref{fig:gt_strength}\hyperref[fig:gt_strength]{(c),(f)}.

The main effect of finite temperature is the appearance of GT strength below the threshold, which is defined by $E^*_t = \pm (\lambda_{n} - \lambda_{p})$ for GT${}^\pm$ transitions.
In the following, we denote such strength as de-excitations or simply, negative energy transitions. At $T = 0.5$~MeV there is almost no contribution of the de-excitations. Overall, the strength function above threshold displays only moderate changes related to the reduction of pairing with increasing temperature. Since the FT-HBCS model is formulated within a grand canonical ensemble, a sharp vanishing of pairing gaps occurs at the critical temperature $T_c$ \cite{Goodman1981}. 
For ${}^{58}$Ni pairing properties vanish at $T_c = 0.74$ MeV. However, strength below the threshold increases with increasing temperature. This is a consequence of the thermal prefactor in Eq.~\eqref{eq:ftpnrqrpa_strength} which allows for a larger number of transitions with excitation energies below the threshold to contribute to the strength function. At $T = 2.0$~MeV the appearance of negative energy strength is clearly seen for both GT${}^+$ and GT${}^-$ transitions. Above the pairing collapse temperature, apart from inducing more transition strength below the threshold, temperature also modifies the strength function above the threshold due to the thermal unblocking, i.e., altering the occupation factors of quasiparticle levels. This effect allows for previously blocked transitions between fully occupied levels to occur.

In Fig.~\ref{fig:gt_strength}\hyperref[fig:gt_strength]{(c),(f)}, the difference between zeroth moments $\Sigma^0(\text{GT}^-) - \Sigma^0(\text{GT}^+) $ reproduces the Ikeda sum rule (red solid line) up to 93\%. It is well known that for zero-temperature PNRQRPA based on relativistic interactions it is necessary
to include the antiparticle-hole contribution to reproduce the sum rules~\cite{PhysRevC.67.034312,PhysRevC.69.054303}. In the present work we omit the antiparticle transitions for simplicity, and the small discrepancy in the sum rule is attributed to these missing contributions which, to a good approximation, can be neglected in charge-exchange calculations. 

The above numerical example demonstrates the need for a consistent treatment of the FT-PN(R)QRPA strength function. To satisfy the Ikeda sum rule using the correct interpretation of the physical strength distribution, which includes the thermal prefactor, it is necessary to include the negative energy transitions.

\subsection{Sum rule with the finite amplitude method}

To complement the calculations in Section~\ref{ssec:rqrpa_sumrule}, we now demonstrate the Ikeda sum rule for ${}^{58}$Ni using the non-relativistic FT-PNQRPA and the charge-changing finite amplitude method (PNFAM)~\cite{Mustonen2014}. The FAM is an efficient means to solve the linear response equations while avoiding the expensive construction of the residual interaction matrix. It accomplishes this by computing the perturbation of the Hamiltonian directly with a finite difference,
\begin{equation}
\begin{aligned}
    \delta \widetilde{H}(\omega) &= \frac{\partial H}{\partial R} \Bigg\rvert_{R=\widetilde{R}_0} \delta \widetilde{R}(\omega) 
    \\&
    =
    \lim_{\eta \rightarrow 0} \frac{1}{\eta} \left[H[ \widetilde{R}_0 + \eta \delta \widetilde{R}(\omega)] - H[ \widetilde{R}_0] \right]
    \\&
    =
    % \begin{pmatrix} 
    (
    \delta \widetilde{H}^{11}, %\\ 
    \delta \widetilde{H}^{20}, %\\
    \delta \widetilde{H}^{02}, %\\
    \delta \widetilde{H}^{\bar{11}}
    )
    % \end{pmatrix}
    \,,
\end{aligned}
\end{equation}
where $\widetilde{R}_0$ is the FT-HFB solution for the generalized density. In terms of the Hamiltonian perturbation, we can rearrange Eq.~\eqref{eq:ft_linear_response} to obtain the FT-FAM equations,
\begin{equation}\label{eq:ft_FAM}
    \left[\mathcal{E} - \omega M\right] \delta \widetilde{R}(\omega) = T \left[\delta \widetilde{H}(\omega) + \mathcal{F}(\omega)\right]
    \,,
\end{equation}
which can be solved for the density response by iteration.
From the FAM response we can then obtain the transition strength function through Eq.~\eqref{eq:ft_fam_strength}. In the charge-changing case, the perturbed Hamiltonian for Skryme functionals without proton-neutron mixing can be evaluated directly with the perturbed density, i.e., $\delta \widetilde{H} = H[\delta \widetilde{R}]$. 

The Ikeda sum rule is computed from the residues of Eq.~\eqref{eq:ft_fam_strength} for the Gamow-Teller external field operator. We obtain a sum of the residues via complex contour integration of the physical strength distribution,
\begin{equation}\label{eq:cci}
    \Sigma^0(\text{GT}^\pm) = \frac{1}{2\pi i} \oint_{C} d \omega\ \frac{\widetilde{S}_{\mathrm{GT}^\pm}(\omega)}{1 - e^{-\beta \omega}}
    \,,
\end{equation}
where we take the contour $C$ to be a circle centered on the real axis.

The thermal prefactor appearing in Eq.~\eqref{eq:cci}, however, is difficult to treat with the complex contour integration method. We need to include strength from excitations at $\omega>0$ and de-excitations at $\omega<0$, but the thermal prefactor contains poles along the imaginary axis at $\omega = n(2 \pi/\beta) i$ for $n=0,1,2,\ldots$ that should not be included in the contour integration. Naively, we might try to avoid these poles by using two contours. We can place one on either side of the imaginary axis and adjust the bounds so they come close to, but do not touch, the imaginary axis itself. In practice, however, if the boundary of a contour passes very near to one of these poles, the contour integration suffers numerical instabilities unless an unfeasibly dense discretization is used in the region near the pole.

\begin{figure}[t]
\caption{\label{fig:contours} Schematic representations of (a) contours and pole at $\omega=0$ for the fractional residue proof (see main text), (b) all contours required for the finite-temperature sum rule, and (c) deforming the contour to avoid poles on the imaginary axis. Poles on the real and imaginary axes are black markers, with circles representing poles from the thermal prefactor and crosses poles from the FT-QRPA strength function.}
\centering
\includegraphics[width=.45\columnwidth]{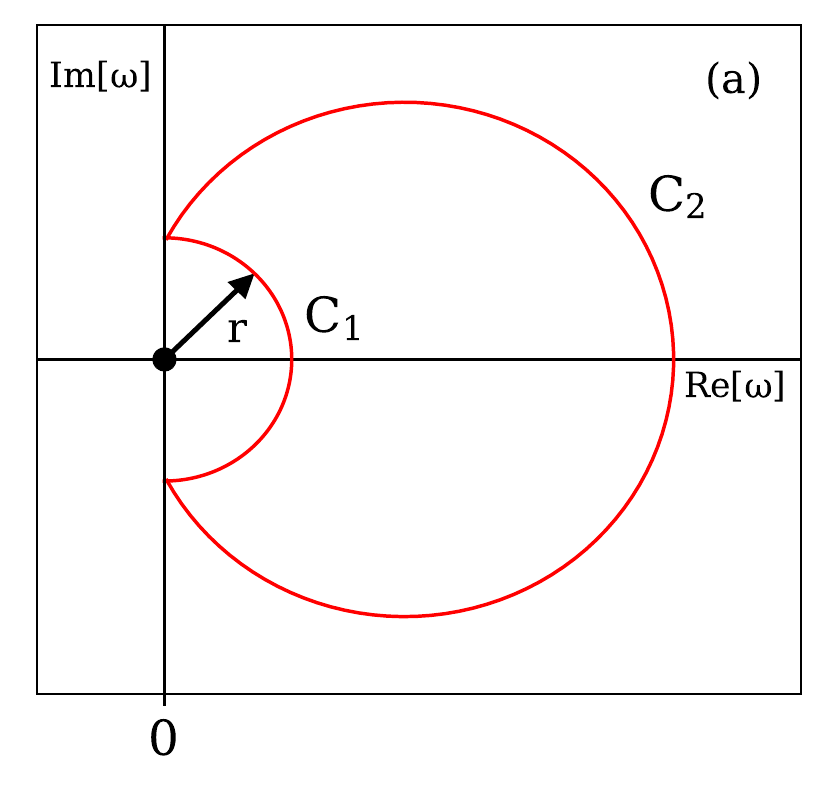}
\includegraphics[width=.45\columnwidth]{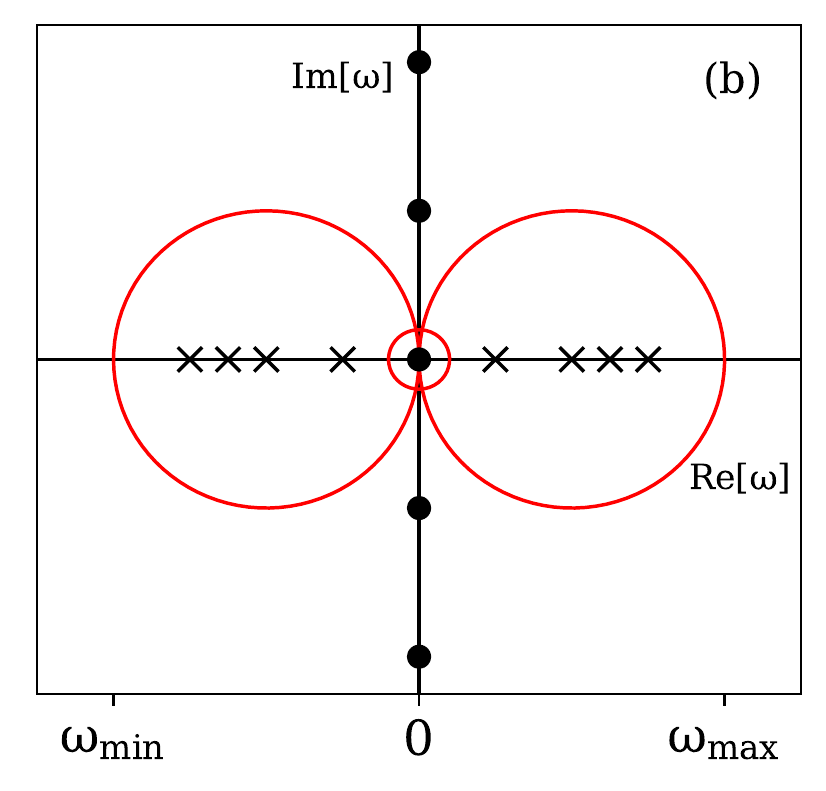}
\includegraphics[width=.45\columnwidth]{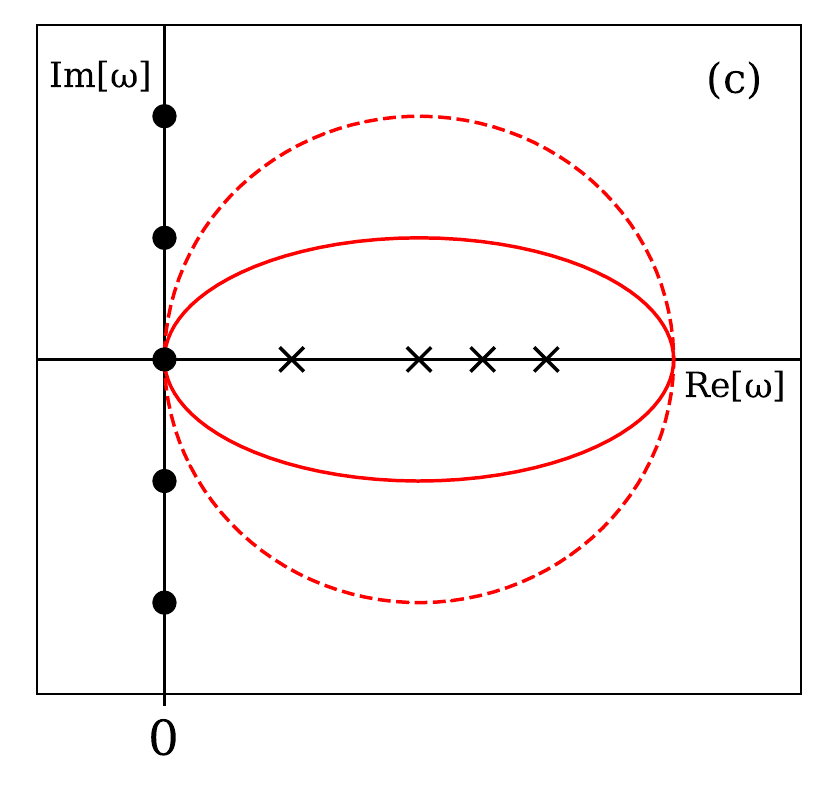}
\end{figure}

\begin{table*}
\centering
\caption{\label{table:fam_sum_rule}
Contributions to the Ikeda sum rule in $^{58}$Ni computed with the FT-PNFAM and the SKO' interaction. Columns labeled \% indicate ${(\Sigma^0(\mathrm{GT}^-) - \Sigma^0(\mathrm{GT}^+))/ 3(N-Z) \times 100\%}$.
}
\renewcommand{\arraystretch}{1.125}
\begin{ruledtabular}
\begin{tabular}{c|d{1.5}d{1.5}d{1.5}|d{1.5}d{1.5}d{1.5}|d{1.5}d{1.5}d{1.5}}
 & \multicolumn{3}{c|}{$\omega < 0$} & \multicolumn{3}{c|}{$\omega > 0$} &  \multicolumn{3}{c}{Total}
\\
 \multicolumn{1}{c|}{T[MeV]}& \multicolumn{1}{c}{$\Sigma^0(\mathrm{GT}^+)$} & \multicolumn{1}{c}{$\Sigma^0(\mathrm{GT}^-)$} & \multicolumn{1}{c|}{\%} &  \multicolumn{1}{c}{$\Sigma^0(\mathrm{GT}^+)$} & \multicolumn{1}{c}{$\Sigma^0(\mathrm{GT}^-)$} & \multicolumn{1}{c|}{\%} & \multicolumn{1}{c}{$\Sigma^0(\mathrm{GT}^+)$} & \multicolumn{1}{c}{$\Sigma^0(\mathrm{GT}^-)$} & \multicolumn{1}{c}{\%}
\\ \hline
 0.0 & 0.0 & 0.0 & 0.0 & 5.87592 & 11.86673 & 99.84677 & 5.87592 & 11.86673 & 99.84677
\\
 0.5 & 0.03442 & 0.00607 &  -0.47251 & 6.38447 & 12.41216 &  100.46141 & 6.41890 & 12.41823 & 99.98890
\\
 1.0 & 0.21849 & 0.08811 &  -2.17301 & 6.25650 & 12.38840 & 102.19831 & 6.47499 & 12.47651 & 100.02530
\\
 1.5 & 0.61711 & 0.19318 &  -7.06545 & 6.18013 & 12.60575 & 107.09364 & 6.79724 & 12.79893 & 100.02819
\\
 2.0 & 1.07395 & 0.45722 & -10.27882 & 6.02681 & 12.64820 & 110.35647 & 7.10076 & 13.10542 & 100.07764
\end{tabular}
\end{ruledtabular}
\end{table*}

Our solution to this problem is to use two contours on either side of the imaginary axis, but place them such that they pass through the pole at $\omega=0$. We find that, so long as no point on the discretized contour lies exactly on $\omega = 0$, the integration yields a stable result with a reasonably dense grid. We can understand this behavior with the concept of fractional residues. One can prove that if $z_0$ is a simple pole of $g(z)$, and $C(r)$ is an arc of the circle defined by $\{r = \abs{z - z_0} \}$ subtended by an arc of angle $\alpha$, then~\cite{Gamelin2001},
\begin{equation}
    \lim\limits_{r \rightarrow 0} \int_{C(r)} g(z)\, dz = \alpha i\, \text{Res}[g(z),z_0]
    \,.
\end{equation}
Thus, if we take $C(r)$ to be $C_1$ in Fig.~\ref{fig:contours}\hyperref[fig:contours]{(a)}, in the limit that $r \rightarrow 0$ (and $C_2$ shifts such that $C_1$+$C_2$ remains a closed contour), we find that the contour integration yields \textit{half} the usual residue, $\lim\limits_{r \rightarrow 0} \int_{C_1(r)+C_2} g(z)\, dz = \pi i\, \text{Res}[g(z),0]$. The inclusion of (a part of) the residue in the contour contributes to the stability of the numerical integration.
Thus, to sum all the strength we can perform two types of integrations: one involving contours that intersect $\omega=0$, and one with a contour around just the pole at $\omega=0$ so that we can subtract (half) its contribution from the former results. The latter integration is numerically stable so long as the pole is centered in the contour, because the value of $\widetilde{S}_F(\omega)/(1 - e^{-\beta \omega})$ will not vary much along the contour. Figure~\ref{fig:contours}\hyperref[fig:contours]{(b)} exemplifies all the contours necessary for a complete calculation. This approach to the sum rule necessarily differs, for example, from the one described in Ref.~\cite{PhysRevC.87.064309} because the thermal prefactor prohibits calculating the strength function near or along the imaginary axis.

While the contour integration method just described works well for high temperatures, for low temperatures the poles off the real axis can also get very close to the boundaries of large contours and induce instabilities in the integration. To address this issue, we deform the contour slightly into an ellipse such that the lowest pole on the imaginary axis is sufficiently far from the edge of the contour, as illustrated in Fig~\ref{fig:contours}\hyperref[fig:contours]{(c)}. We find that maintaining a distance of $2\pi/\beta$ is sufficient. For the smallest temperatures, this approach requires deforming the contour so much that it gets very close to the real axis everywhere, which is also undesirable. In these cases, however, the thermal prefactor is effectively a unit step function and de-excitations are negligible. We therefore neglect de-excitations and exclude the prefactor for temperatures at or below $1.0~\mathrm{GK} \approx 0.1~\mathrm{MeV}$.

Applying this procedure for summing the strength, we use the axially-deformed Skyrme PNFAM to calculate the Ikeda sum rule in $^{58}$Ni for $T=0$--2~MeV. We use the Skyrme functional SKO' fit for the global calculations of Refs.~\cite{Mustonen2016,Ney2020} with an effective axial vector coupling $g_A = 1.0$. We included QRPA energies from $\omega = -50~\mathrm{MeV}$ to $+50~\mathrm{MeV}$, and 88 point Gauss-Legendre quadratures were performed for all contour integrations. Our results, which appear in Table~\ref{table:fam_sum_rule}, are within 0.1\% of the exact sum rule for all temperatures studied. The very small deviation is likely attributed to a small amount of missing strength beyond the energy range considered and the buildup of small numerical errors. Although we do not gain information about the strength distribution from the contour integration, we can isolate the relative contributions from excitations ($\omega > 0$) and de-excitations ($\omega < 0$) to the sum rule.  We see from Table~\ref{table:fam_sum_rule} that de-excitations become increasingly important with temperature, accounting for more than 10\% of the sum rule at a temperature of 2.0~MeV.

\section{Stellar electron capture rates}\label{sec:ec_rates}

To exemplify how our presentation of the FT-QRPA impacts nuclear decays, we compute stellar EC rates. 
The stellar environment just prior to the supernova explosion (presupernova) is characterized by a high temperature $T$, and a product of baryon density $\rho$ and electron-to-baryon ratio $Y_e$ ($\rho Y_e$). Atoms are assumed to be fully ionized and the electron gas is described by a Fermi-Dirac distribution. Under the extreme presupernova conditions, nuclei can be found in highly-excited states. Decays from such states, characterized by negative $Q$-values, are known as de-excitations~\cite{PhysRevC.81.015804,PhysRevC.100.025801}. Therefore, the FT-QRPA and its inclusion of transitions from thermally populated excited states is well suited to describe stellar weak-interaction rates.
In this section, we investigate stellar EC rates using both relativistic and non-relativistic models from Section~\ref{sec:bgt}. Calculations are performed for ${}^{58}$Ni and ${}^{78}$Ni which are known to be of importance for the dynamics of CCSNe~\cite{Sullivan_2015,PhysRevC.83.044619}.

\begin{figure*}
\centering
\includegraphics[width = \linewidth]{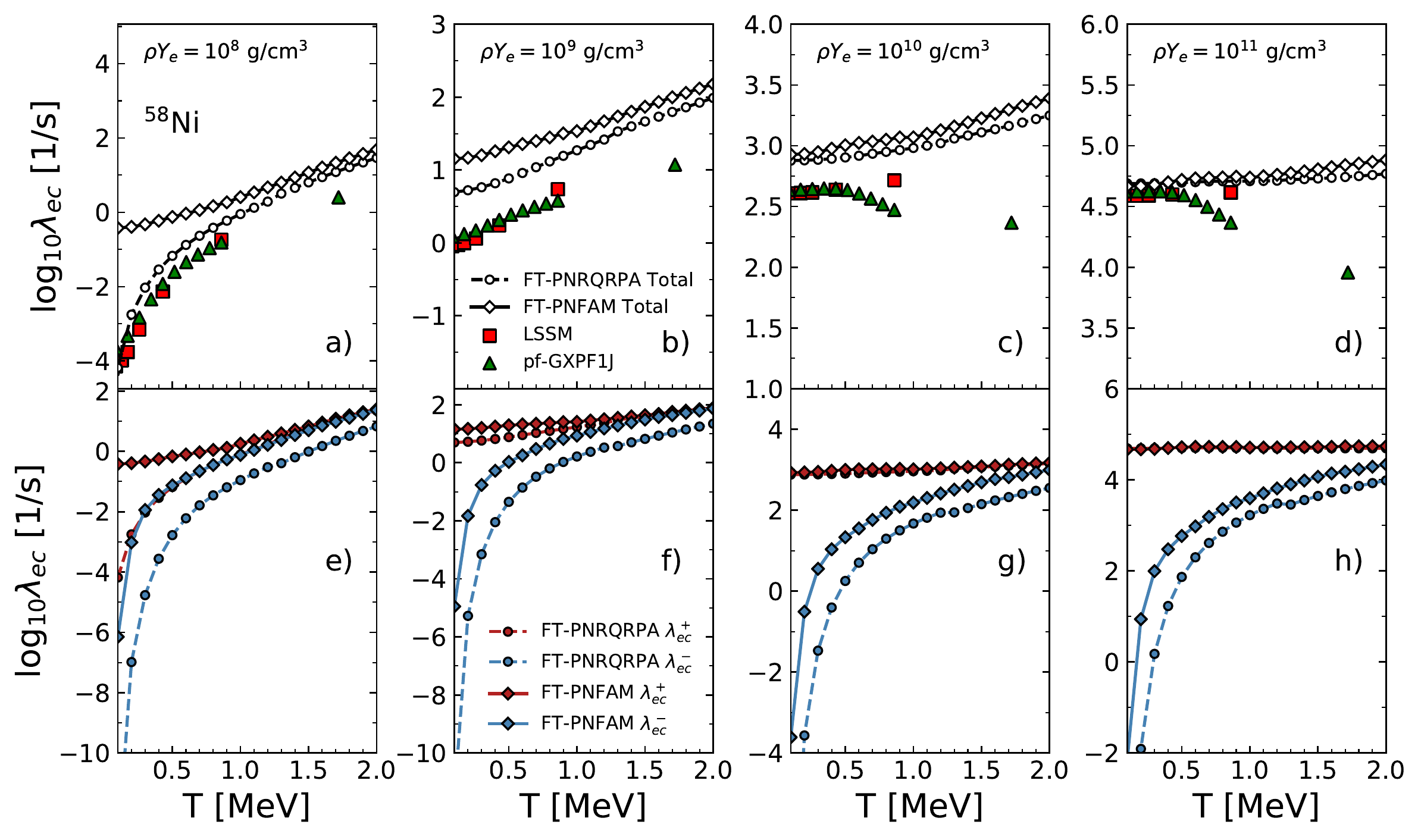}
\caption{(a)--(d) the temperature dependence of the total electron capture (EC) rate $\lambda_{ec}$ in ${}^{58}$Ni in the range $T = 0$--$2$ MeV at stellar densities $\rho Y_e = 10^8$--$10^{11}$ g/cm${}^3$. Results are calculated with the non-relativistic FT-PNFAM (diamonds) and relativistic FT-PNQRPA (circles), and compared with the large scale shell-model calculations (red squares)~\cite{LANGANKE2000481} and shell-model calculations based on the pf-GXPF1J interaction (green triangles)~\cite{PhysRevC.83.044619,Mori_2016}. (e)--(h) contribution of the rate from excitations $\lambda_{ec}^+$ (red) and from de-excitations $\lambda_{ec}^-$ (blue) to the total EC rate as calculated with the FT-PNFAM (diamonds) and FT-PNRQRPA (circles).}\label{fig:ec1}
\end{figure*}

Here we provide a brief outline of the EC rate calculation, with details appearing in Ref.~\cite{PhysRevC.105.055801}. In this work we assume the allowed GT approximation, thus stellar EC rates can be written as
\begin{equation}
\lambda_{ec} = \frac{\text{ln}2}{\kappa} \frac{1}{Z} \sum \limits_{i,f} e^{-\beta \omega_i} |\langle f | \boldsymbol{\sigma} \tau_+ | i \rangle|^2 f(W_0^{i,f}, \mu),
\end{equation}
where $| i (f) \rangle$ are the exact initial (final) nuclear states with energy $\omega_{i (f)}$ and angular momenta $J_{i (f)}$, $\kappa = 6147$~s is the decay constant, and $Z = \sum_i (2J_i+1) e^{-\beta \omega_i}$ is the partition function. The dimensionless phase space factor $f(W_0^{i,f},\mu)$ is an integral over electron energies $W = E_e/ m_e c^2$, where $m_e$ is the electron mass. The energy $W_0^{(i,f)}$ is the maximum electron energy for the transition from $\ket{i}$ to $\ket{f}$ in the $\beta^+$ direction. The phase space integrand is folded with a Fermi-Dirac distribution of electrons, which depends on the electron chemical potential $\mu$
% \begin{equation}
% f_e(W) = \left[ 1 + \exp\left( \frac{W- \mu/(m_e c^2)}{k_B T} \right) \right]^{-1},
% \end{equation}
% where $\mu$ is the chemical potential of electrons 
determined from the charge-neutrality condition for a given stellar density $\rho Y_e$~\cite{PhysRevC.105.055801}. 

In the FT-QRPA, the EC rate can be expressed in terms of residues of the strength function
\begin{equation}\label{eq:ec_rate_gt}
\lambda_{ec} = \frac{\text{ln}2}{\kappa} \sum \limits_{k\pm} \text{Res} \left[ \frac{\tilde{S}_{\mathrm{GT}^+}(\omega)}{1-e^{-\beta \omega}}f(W_0[\omega]),\ \Omega_k \right],
\end{equation}
where the summation is performed over FT-QRPA modes with both positive and negative eigenvalues and $W_0[\omega]$ is the threshold energy expressed in terms of the FT-QRPA perturbing energy $\omega$~\cite{PhysRevC.105.055801}. From the PNFAM, the sum of residues can be obtained via complex contour integration, while from the matrix theory the residues are computed from the FT-PNQRPA eigenvectors.
In principle, at finite temperature the excitation energies can take values $-\infty < \omega < \infty$. However, the prefactor $(1-e^{-\beta \omega})^{-1}$ provides a cut-off for large negative energies, and the phase space integral provides a cut-off for large positive energies. To illustrate the contribution of de-excitations, we separate the rate into $\lambda_{ec} = \lambda^+_{ec} + \lambda^-_{ec}$, where $\lambda^+_{ec} = \lambda_{ec} \big\lvert_{\omega>0}$ is the rate from excitations and $\lambda^-_{ec} = \lambda_{ec} \big\lvert_{\omega<0}$ is the rate from de-excitations.
% By inserting the FT-QRPA strength from Eq.~\eqref{eq:qrpa_ft_strength} in the above expression, the total EC rate can be separated into two terms
% \begin{equation}
% \lambda_{ec} = \lambda_{ec}^+ + \lambda_{ec}^-,
% \end{equation}
% with the definitions
% \begin{equation}
% \lambda_{ec}^+ = \frac{\text{ln}2}{\kappa} g_A^2 \sum_k \frac{|\langle [\Gamma^k, \boldsymbol{\sigma}\tau_+] \rangle|^2}{1-e^{-\beta \Omega_k}} f(W_0^k[\Omega_k]),
% \end{equation}
% \begin{equation}\label{eq:dexc_rate}
% \lambda_{ec}^- = \frac{\text{ln}2}{\kappa} g_A^2 \sum_k \frac{|\langle [\Gamma^k, \boldsymbol{\sigma}\tau_-] \rangle|^2}{1-e^{\beta \Omega_k }} f(W_0^k[-\Omega_k]),
% \end{equation}
% where $\lambda_{ec}^+$ is the EC rate due to excitations and $\lambda_{ec}^-$ is the EC rate due to de-excitations. 
The axial-vector coupling constant $g_A$ is quenched from the free nucleon value to $g_A = 1.0$ in both calculations, consistent with previous works on EC \cite{PhysRevC.105.055801,PhysRevC.102.065804,PhysRevC.83.045807,PhysRevC.80.055801}.

\begin{figure*}
\centering
\includegraphics[width = \linewidth]{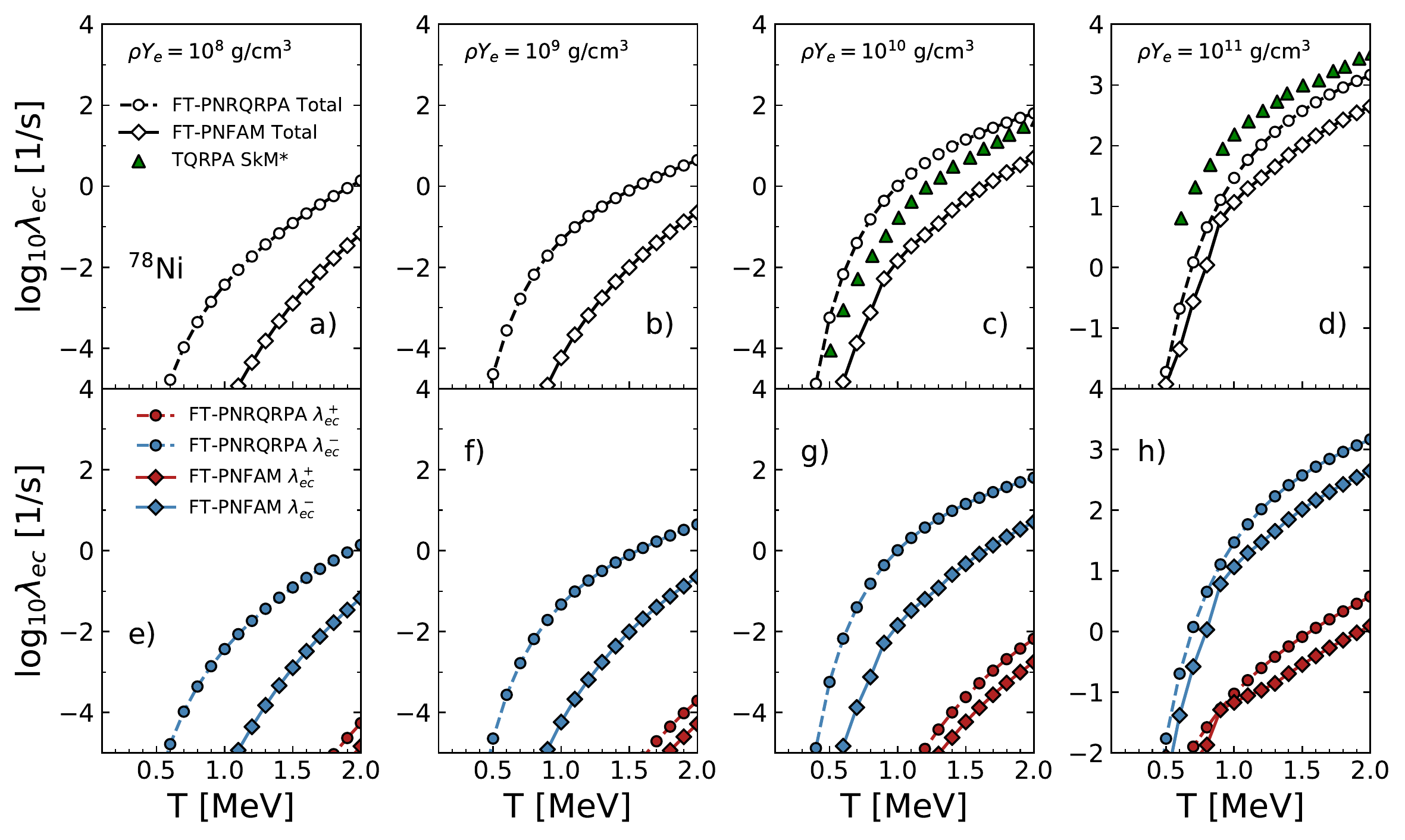}
\caption{The same as in Fig.~\ref{fig:ec1} but for ${}^{78}$Ni. Results are compared with the TQRPA calculations based on the non-relativistic SkM${}^*$ interaction in Ref.~\cite{PhysRevC.100.025801} (green triangles).}\label{fig:ec2}
\end{figure*}

% While the non relativistic FT-PNFAM calculation uses Eq.~\eqref{eq:ec_rate_gt} to calculate the EC rates, the relativistic FT-PNRQRPA calculation employs the formalism of Walecka et al. developed in Refs.~\cite{PhysRevC.6.719,walecka2004theoretical} which also takes into account the momentum-transfer dependence of the multipole transition operators. However, for electron energies up to 40~MeV the transverse electric operator $\hat{\mathcal{T}}_J^{el.}$~\cite{PhysRevC.6.719,walecka2004theoretical} for $J^\pi = 1^+$ reduces to the well-known GT form $\boldsymbol{\sigma} \tau_+$ \cite{PhysRevC.100.025801}. Therefore, differences between two formulations of weak interaction are negligible for considerations within this work.

In Fig.~\ref{fig:ec1}\hyperref[fig:ec1]{(a)--(d)} we show the temperature dependence of the allowed $J^\pi = 1^+$ EC rate $(\lambda_{ec}$) in ${}^{58}$Ni at stellar densities in the range $\rho Y_e = 10^8$--$10^{11}$~g/cm${}^3$. Results are displayed for both the relativistic FT-PNQRPA and non-relativistic FT-PNFAM calculations, and compared with the large scale shell-model calculations in Ref.~\cite{LANGANKE2000481} and the shell-model calculations based on the pf-GXPF1J interaction~\cite{PhysRevC.83.044619,Mori_2016}. It is observed that EC rates in ${}^{58}$Ni tend to increase with increasing temperature and increasing stellar density. The former is a consequence of the so-called thermal unblocking effect, where correlations due to finite temperature allow for previously blocked GT transitions (the contribution of de-excitations can be neglected for ${}^{58}$Ni), while the latter stems from the fact that higher density implies higher electron chemical potential, thus allowing for more strength to contribute to the overall rate. 

At the lower stellar densities of $10^8$--$10^9$ g/cm${}^3$ there are substantial differences in the total EC rate as calculated with the FT-PNFAM and FT-PNRQRPA. This is related to employing different model interactions. Namely, the FT-PNFAM employs the non-relativistic Skyrme SkO' interaction and allows for axial deformations, while the relativistic FT-PNQRPA is based on the D3C${}^*$ interaction and is restricted to spherical configurations. This leads to systematic variance between two calculations since different effective interactions predict different quasiparticle bases.
In particular, the FT-PNFAM calculation finds the potential energy surface of ${}^{58}$Ni to be extremely shallow, with an oblate local minimum providing the lowest energy solution. Additionally, this axially-deformed ground state has a non-zero EC $Q$-value, explaining why the rate does not tend to zero at low temperatures and densities.

However, these differences become less important above the critical temperature for shape transition, where small deformation effects in ${}^{58}$Ni get washed out~\cite{Egido_1993}. Indeed, we observe differences between the model calculations decreasing with increasing temperature in Fig.~\ref{fig:ec1}. For higher stellar densities ($10^{10}$--$10^{11}$ g/cm${}^3$) the EC rates are almost independent of the particular details in the GT strength function, implying that only the overall GT strength matters \cite{Langanke_2021,PhysRevC.105.055801}. Hence, the variance between both model calculations decreases with increasing $\rho Y_e$. A fair agreement of the total EC rate, as calculated with our FT-QRPA models, is obtained with respective shell-model calculations, although both FT-QRPA calculations tend to overestimate the shell-model EC rates.
% now describe de-excitations

In Fig.~\ref{fig:ec1}\hyperref[fig:ec1]{(e)--(h)} the total EC rate in ${}^{58}$Ni is decomposed into contributions from excitations ($\lambda_{ec}^+$) and de-excitations ($\lambda_{ec}^-$). It is observed that the rate due to de-excitations increases considerably with temperature. However, its overall impact on the total rate ($\lambda_{ec}$) is negligible up to $T = 2$~MeV. A similar trend is confirmed for both model calculations in this work.

To observe how the contribution of de-excitations can lead to significant changes in the EC rates, in Fig. \ref{fig:ec2}\hyperref[fig:ec2]{(a)--(d)} we display the total EC rate in ${}^{78}$Ni for stellar densities in the range $10^8$--$10^{11}$ g/cm${}^3$, with decomposition into the excitation and de-excitation rates in Fig.~\ref{fig:ec2}\hyperref[fig:ec2]{(e)--(h)}. Clearly, the contribution of de-excitations dominates the EC rates for all stellar densities, starting already at $T \approx 0.5$ MeV for both model calculations. It is larger than the respective excitation contribution by more than a few orders of magnitude. There is again variation between the FT-PNFAM and FT-PNRQRPA rates, but the trend is observed in both. Since ${}^{78}$Ni is predicted as doubly-magic (and thus spherical) in both calculations, the variation is a consequence of different effective interactions. In Fig.~\ref{fig:ec2}\hyperref[fig:ec2]{(c)--(d)} we have also displayed the data from Ref.~\cite{PhysRevC.100.025801} based on the TQRPA with the non-relativistic SkM${}^*$ interaction (green triangles). With the proper interpretation of the FT-QRPA strength function, as clarified in this work, we can reproduce the TQRPA rates with both the FT-PNRQRPA and FT-PNFAM. Otherwise, both FT-QRPA rates would considerably underestimate the TQRPA rates. Again, note that all three calculations use different effective interactions, thus agreement between the different models, FT-QRPA using either matrix FT-PNRQRPA or FT-PNFAM formulation, and the TQRPA is quite satisfying.

To explain why de-excitations dominate in ${}^{78}$Ni but contribute much less in ${}^{58}$Ni, we note that ${}^{58}$Ni is located near the valley of stability, with a considerable amount of the GT${}^+$ strength (cf.~Table~\ref{table:fam_sum_rule}). Therefore, the EC rates are going to have a significant contribution from the GT${}^+$ strength even at low temperatures. However, due to its neutron excess, GT${}^+$ strength in ${}^{78}$Ni is substantially suppressed. 
%This is because of filled neutron shells and a high neutron Fermi level, where $p \to n$ transitions are effectively blocked. 
On the other hand, because of its large GT${}^-$ strength, with increasing temperature the thermal prefactor in Eq.~\eqref{eq:ec_rate_gt} allows for more strength to contribute to the total EC rate. This sudden jump in the de-excitation contribution at $\rho Y_e = 10^8$ g/cm${}^3$ occurs around $T = 0.5$ MeV, and at somewhat smaller temperatures at higher densities. Although with increasing stellar density the rate from excitations also increases, even at $\rho Y_e = 10^{11}$~g/cm${}^3$ the de-excitation contribution for $T > 0.5$~MeV is larger by a few orders of magnitude.

\section{Conclusions}

In this work, we presented a detailed investigation of the FT-QRPA derived in terms of thermal averages over the FT-HFB (or FT-HFBCS) statistical ensemble. We demonstrated connections between different formulations of the FT-QRPA linear response equations used in the literature and showed that all physical quantities are independent of the formulation used. Furthermore, we illustrated how the temperature-dependent metric formulation presents itself as the natural choice for the FT-QRPA equation because its properties mirror those of the zero-temperature QRPA. We then elucidated several properties of the FT-QRPA strength function, showing that the thermal prefactor and strength from both excitations and de-excitations are essential. We also located individual transitions in the FT-QRPA strength function, demonstrating that the strength at a given FT-QRPA energy approximates a thermal average over all transitions with a fixed transition energy.
% At zero-temperature, the exact strength function is defined only for positive excitation energies $\omega > 0$, and the physical strength function is related to its imaginary part. On the other hand, at finite-temperature, de-excitation transitions from excited states to final states lower in excitation energy are allowed. The finite-temperature physical strength distribution is not equal to the imaginary part of the exact finite-temperature strength function because the contribution of reverse transitions governed by $\hat{F}^\dag$ is mixed with the forward transitions governed by $\hat{F}$. We showed that the exact strength function can be expressed in terms of $\hat{F}$-induced transitions, multiplied with the thermal prefactor as in Eq.~\eqref{eq:ft_prefactor}. This implies that a part of the physical strength at finite temperature, in addition to excitations with $\omega > 0$, originates also from de-excitations with $\omega < 0$. Finally, by connecting the exact response function to the FT-QRPA response we conclude that de-excitations should be also included in the FT-QRPA. 

To illustrate the correctness of our discussion, we employed two recently developed models in the charge-exchange channel: the non-relativistic, axially deformed FT-PNFAM based on the Skyrme EDFs and the relativistic, spherical FT-PNQRPA based on the meson-exchange D3C${}^*$ EDF. We verified analytically, and then numerically with both models for the case of $^{58}$Ni, that the Ikeda sum rule is satisfied with the correct treatment of the FT-QRPA strength function. As a physical application, we then computed stellar EC rates for $^{58,78}$Ni for a range of stellar densities and temperatures.
% We performed calculations of the temperature evolution of GT strength in ${}^{58}$Ni up to 2 MeV. De-excitations start to contribute to the overall strength distributions for $T > 1$ MeV. We have also verified, both analytically and numerically, that our formulation of the FT-QRPA satisfies the Ikeda sum rule. Lastly, we present results for EC rates in ${}^{58,78}$Ni for a range of stellar densities $\rho Y_e = 10^8$--$10^{11}$ g/cm${}^3$. 
Although de-excitations play almost no role in the EC rate for ${}^{58}$Ni up to $T = 2$ MeV, they dominate the EC rate in ${}^{78}$Ni starting already from $T \approx 0.5$ MeV at $\rho Y_e = 10^8$ g/cm${}^3$ because of its large negative $Q$-value. Only with the inclusion of de-excitations are we able to obtain reasonable agreement between our calculations and the EC rates based on the non-relativistic TQRPA calculations in Ref.~\cite{PhysRevC.100.025801}. Similar trends are confirmed for both relativistic and non-relativistic FT-QRPA calculations in this work.

Recently, in Ref.~\cite{PhysRevC.105.055801} we have demonstrated that present models produce consistent results for EC rates of nuclei in the $N = 50$ region, thus yielding small uncertainties for main CCSNe observables. This gives us confidence that the main correlations necessary for the description of EC rates are well enveloped, at least under extreme stellar conditions. Applying present models to other astrophysical scenarios, as well as to systematic calculation of EC rates across the nuclear chart, remain important tasks that should be pursued to further constrain astrophysical uncertainties.

\begin{acknowledgments}
This work was supported by the US National Science Foundation under Grant PHY-1927130 (AccelNet-WOU: International Research Network for Nuclear Astrophysics [IReNA]). A.R. and N.P. acknowledge support by the QuantiXLie Centre of Excellence, a project co-financed by the Croatian Government and European Union through the European Regional Development Fund, the Competitiveness and Cohesion Operational Programme (KK.01.1.1.01.0004).
\end{acknowledgments}

\appendix
\section{FT-QRPA proofs}\label{sec:appa}

In this section we enumerate several properties of the FT-QRPA matrix for the class of eigenvalue problems originating from Eq.~\eqref{eq:ft_linear_response_p},
\begin{equation}\label{eq:p_dept_FTQRPA}
    (M \tsp)\ \delta \widetilde{R}_p^k = \Omega^k\ \delta \widetilde{R}_p^k
    \,.
\end{equation}
Similar proofs for the zero-temperature RPA matrix are well-known, see for example Refs.~\cite{Ring2004} and~\cite{Ripka1986}. Here we extend these proofs to the $p$-dependent FT-QRPA eigenvalue problems and show they largely still apply in the same way as in the zero-temperature case, with a few caveats related to the temperature dependence. Moreover, we show that several key quantities are independent of $p$, demonstrating the equivalence of all formulations. 

All of the proofs that follow require that the two-quasiparticle basis is constructed such that $E_\alpha > E_\beta$ and therefore $T$ is real and positive definite. Additionally, for the sake of simplicity we assume $\Omega^k \neq 0$.

\begin{property}\label{prop:LR_EVP}
Given the right eigenvalue problem in Eq.~\eqref{eq:p_dept_FTQRPA}, the matrix $M\tsp$ has a left eigenvector ${\tlp = (\trr_{1-p}^k)^\dagger M}$ with eigenvalue $\Omega^k_{(L)} = \Omega^{k*}$.
\end{property}
\begin{proof}
By definition,
\begin{equation}\label{eq:left_FTQRPA}
    \tlp^k (M\tsp) = \Omega^k_{(L)} \tlp^k
    \,.
\end{equation}
Using the claim for $\tlp$ as an ansatz, the transpose of 
Eq.~\eqref{eq:left_FTQRPA} states,
\begin{equation}
\begin{aligned}
    (M\tsp)^T (\tlp^{k})^{T}
    =(M M \ts_{1-p}\ \trr_{1-p}^k)^*
    =\Omega^{k*} (\tlp^{k})^{T}
    \,.
\end{aligned}
\end{equation}
The second equality follows from the properties $M^2 = 1$, and if $T>0$, then $T^{1-p}$ is real and $\tsp^\dagger = \ts_{1-p}$.
\end{proof}

\begin{property}\label{prop:orthogonal}
Eigenvectors of $M\tsp$ are orthogonal if $\Omega^k \neq \Omega^{l*}$. 
\end{property}
\begin{proof}
For the non-hermitian eigenvalue problem, the norm is defined as the scalar product of left and right eigenvectors, and is independent of $p$,
\begin{equation}
\begin{aligned}
    \tlp^k \trp^l = (\delta R^{k})^{\dagger} M T\ \delta R^l
    \,.
\end{aligned}
\end{equation}
Now consider the left and right eigenvalue problems for any two solutions. Let us take the difference between
\begin{equation}
    \tlp^l \left[(M \tsp) \trp^k = \Omega^k \trp^k\right]
    \,,
\end{equation}
and
\begin{equation}
    \left[ \tlp^l (M \tsp) = \Omega^l_{(L)} \tlp^l\right] \trp^k
    \,.
\end{equation}
The $p$-dependence on the left-hand-side disappears, giving $(\delta R^k)^\dagger \ts_M \delta R^l$ for both equations, which cancels after taking the difference. We are left with,
\begin{equation}\label{eq:orthonorm_proof}
    0 = (\Omega^k - \Omega^{l*}) (\delta R^{l })^\dagger MT \delta R^k
\end{equation}
where we have used that $\Omega^l_{(L)} = \Omega^{l*}$ from Property~\ref{prop:LR_EVP}.
So long as $\Omega^k \neq \Omega^{l*}$, the eigenvectors are orthogonal. This result is independent of $p$.
\end{proof}

\begin{corollary}\label{cor:complex_norm}
Eigenvectors belonging to complex eigenvalues have zero norm.
\end{corollary}
\begin{proof}
If $\Im[\Omega^k] \neq 0$, then  $\Omega^k - \Omega^{k*} \neq 0$, and Eq.~\eqref{eq:orthonorm_proof} requires the norm to be zero.
\end{proof}

\begin{property}\label{prop:stability}
If the Hermitian matrix $\ts_M = T \mathcal{H} T + \mathcal{E} T$ has all positive eigenvalues, the eigenvalues of $M\tsp$ are real and the sign of the eigenvalue matches the sign of the norm.
\end{property}
\begin{proof}
Multiplying the right eigenvalue problem for $M\tsp$ on the left by the corresponding left eigenvector leads to
\begin{equation}\label{eq:SM_proof}
\begin{aligned}
    (\delta R^k)^\dagger \ts_M \delta R^k
    &= \Omega^k  (\delta R^{k})^\dagger MT \delta R^k
\end{aligned}
\end{equation}
The right-hand side is the eigenvalue times the norm and the left-hand side is an expectation value of a Hermitian operator. To distinguish solutions of $\ts_M$ in the present context, we denote them with a subscript $(M)$. Using the completeness relation for the solutions of $\ts_M$ we can show,
\begin{equation}\label{eq:herm_expectation}
\begin{aligned}
    (\delta R^k)^\dagger \ts_M \delta R^k
    = \sum\limits_{n} \Omega^n_{(M)} \left\lvert (\delta R^k)^\dagger \delta {R}_{(M)}^n \right\rvert^2
    \,.
\end{aligned}
\end{equation}
Since $\ts_M$ is Hermitian, $\Omega_{(M)}^n$ are real and Eq.~\eqref{eq:herm_expectation} proves the expectation value on the left-hand side is real. Inserting this result into Eq.~\eqref{eq:SM_proof}, we can make the following conclusions:
\begin{enumerate}
    \item If $\Im[\Omega^k] \neq 0$, by Corollary~\ref{cor:complex_norm}, the right-hand side of Eq.~\eqref{eq:SM_proof} vanishes, and therefore $\ts_M$ cannot have all positive eigenvalues.
    \item If all of the eigenvalues of $\ts_M$ are positive, the eigenvalues of $M\tsp$ are real and ${\text{sign}[\Omega^k] = \text{sign}[ (\delta R^{k})^{\dagger} MT \delta R^k]}$.
\end{enumerate}
% We can prove the second claim applies to $\ts_M$ itself by multiplying its right eigenvalue problem by $(\delta R^n_{(M)})^\dagger \widetilde{M}$. This gives,
% \begin{equation}
%     (\delta R^n_{(M)})^\dagger \ts_M \delta R^n_{(M)} = \Omega^n_{(M)} (\delta R^n_{(M)})^\dagger M T \delta R^n_{(M)}
%     \,,
% \end{equation}
% and the property follows immediately.
\end{proof}

\begin{property}\label{prop:UD_EVP}
The right eigensolutions of $M\tsp$ come in two sets. For every solution $\{\Omega^k, \trp^k\}$, there is an orthogonal solution $\{-\Omega^{k*}, \gamma\ \trp^{k*} \}$.
\end{property}
\begin{proof}
The FT-QRPA matrix obeys the relation,
\begin{equation}
    ( M\tsp )^* = - \gamma\, (M\tsp)\, \gamma
    ,\quad
    \gamma \equiv
    \begin{pmatrix}
        0&0&0&1 \\ 0&0&1&0 \\ 0&1&0&0 \\ 1&0&0&0
    \end{pmatrix}
    \,.
\end{equation}
From the complex conjugate eigenvalue problem, we can deduce,
\begin{equation}
\begin{aligned}
    ( M\tsp \trp^k &= \Omega^k \trp^k )^*
    \\
    - \gamma M\tsp \gamma\ \trp^{k*} &= \Omega^{k*} \trp^{k*}
    \\
    M\tsp (\gamma \trp^{k*}) &= - \Omega^{k*} ( \gamma  \trp^{k*} )
    \,.
\end{aligned}
\end{equation}
% Let us denote the solutions $\{\Omega^k, \trp^k\}$ as belonging to set~``A'' and $\{-\Omega^{k*}, \gamma \trp^{k*} \}$ belonging to set~``B''. 
By Property~\ref{prop:orthogonal}, the eigenvectors $\gamma \trp^{k*}$ are orthogonal to $\trp^k$.
% \begin{equation}
% \begin{aligned}
%     0 &= (\Omega^k_{(A)} - \Omega^k_{(B, L)}) (\delta R^{k})^{\dagger} MT (\gamma \trp^{k*})
%     \\
%     &= (\Omega^k_{(A)} - \Omega^{k*}_{(B)}) (\delta R^{k})^{\dagger} MT (\gamma \trp^{k*})
%     \\
%     &= 2 \Omega^k_{(A)} (\delta R^{k})^{\dagger} MT (\gamma \trp^{k*})
% \end{aligned}
% \end{equation}
We may call this result an upper-lower duality (by virtue of $\gamma$ swapping the upper and lower components of $\delta R^k$), while Property~\ref{prop:LR_EVP} defines a left-right duality. We can therefore summarize the FT-QRPA solutions accordingly:
\begin{table}[h!]
% \caption{\label{table:nuclei}}
\renewcommand{\arraystretch}{1.125}
\begin{ruledtabular}
\begin{tabular}{clllc}
     &   & Upper  & Lower  & \\
    \hline \noalign{\vskip 1ex}  
     & Right & $\{\Omega^k,\ \trp^k \}$ & $\{-\Omega^{k*},\ \gamma \trp^{k*} \}$ &\\
     & Left & $\{\Omega^{k*},\ (\trr_{1-p}^k)^\dagger M \}$ & $\{-\Omega^{k},\  (\trr^k_{1-p})^T \gamma M \}$ &
\end{tabular}
\end{ruledtabular}
\end{table}
\end{proof}

\begin{property}\label{prop:linearly_indept}
If all the eigenvalues of $M\tsp$ are real, the eigenvectors form a linearly independent complete set that obeys the closure relation in Eq.~\eqref{eq:closure2}.
\end{property}
\begin{proof}
The solutions are linearly independent if an only if
\begin{equation}
    \sum\limits_k c^k \trp^k = 0
    \,,
\end{equation}
is satisfied by $c^k=0\ \  \forall\ \ k$. Multiplying by any left eigenvector, the orthogonality relation from Property~\ref{prop:orthogonal} gives
\begin{equation}
    \tlp^l \sum\limits_k c^k \trp^k = c^l\  (\delta R^l)^\dagger\, T M\, \delta R^l
    \,.
\end{equation}
If all the eigenvalues are real, then all norms are non-zero, linear independence is satisfied, and the number of linearly independent vectors is equal to the dimension of $M\tsp$, forming a complete set. Furthermore, the expression,
\begin{equation}\label{eq:closure1}
    \sum\limits_k \trp^k\, \tlp^k
\end{equation}
applied to any vector in the complete set returns the same vector and is therefore equal to unity, proving the closure relation. 

We can expand Eq.~\eqref{eq:closure1} to write the closure relation in terms of the upper and lower dual eigenvectors. If we take all positive eigenvalues to be in the set $\{\Omega^k,\ \trp^k\}$, by Property~\ref{prop:stability} they have all have positive norms. Accordingly, the set $\{-\Omega^{k*},\ \gamma\, \trp^{k*}\}$ has all negative eigenvalues and negative norms. We can therefore normalize the upper and lower dual vectors to $+1$ and $-1$, respectively, and can write the closure relation as,
\begin{equation}\label{eq:closure2}
        \sum\limits_{k > 0} \trp^k\, \tlp^k - (\gamma \trp^{k*}) ((\trr_{1-p}^k)^T \gamma M) = 1
        \,.
\end{equation}
where the sum is over modes with positive norm.
\end{proof}

\begin{property}
The $p$-dependent solutions of $M\tsp$ can be used to calculate the same physical transition amplitude.
\end{property}

\begin{proof}
The physical transition amplitude contains exactly one factor of $T$ as a result of taking the ensemble average. We can demonstrate this using the equations of motion of the FT-QRPA~\cite{Sommermann1983}. For excitation operator,
\begin{equation}
        \Gamma^{k \dagger} = \sum\limits_{\mu > \nu} 
        P^k_{\mu\nu} b^\dagger_\mu b_\nu
        +X^k_{\mu\nu} b^\dagger_\mu b^\dagger_\nu
        -Y^k_{\mu\nu} b_\nu b_\mu
        -Q^k_{\mu\nu} b^\dagger_\nu b_\mu
        \,,
\end{equation}
they read,
\begin{equation}
    \left\langle \left[\delta \Gamma^k, \left[H,\Gamma^{k \dagger}\right]\right] \right\rangle = \Omega^k \left\langle \left[ \delta \Gamma^k, \Gamma^{k \dagger} \right] \right\rangle
    \,.
\end{equation}
This implies the amplitudes correspond to thermal averages as given in Eq. (\ref{eq:FTamplitudes}), so the ensemble averaged transition amplitudes are
\begin{equation}\label{F_ensemble_avg_A1}
    \left\langle \left[ \hat{F}, \hat{\Gamma}^{k \dagger} \right] \right\rangle 
    = \mathcal{F}^\dagger\, T\, \delta R^k
    \,.
\end{equation}
Therefore, to get the physical transition amplitude from the eigenvectors of $M\tsp$, we simply need to trace with the appropriate factor of $T$, i.e.,
\begin{equation}\label{F_ensemble_avg_A2}
    \left\langle \left[ \hat{F}, \hat{\Gamma}^{k \dagger} \right] \right\rangle 
    =
    (T^{1-p}\mathcal{F})^\dagger \trp^k
    \,.
\end{equation}
\end{proof}

\section{Derivation of the Ikeda sum rule within the FT-H(F)B or FT-H(F)BCS}\label{sec:appb}

In the proton-neutron single-particle basis the external field operator can be written as
\begin{equation}
\hat{F} = \sum \limits_{pn} F_{pn} c_p^\dag c_n,
\end{equation}
where $F_{pn}$ is the matrix element of the external field operator and $c^\dag_{p(n)}, c_{p(n)}$ proton (neutron) creation and annihilation operators respectively. To write the above expression in the quasiparticle basis we use the Bogoliubov transformation \cite{Ring2004}
\begin{equation}
c_p^\dag = \sum \limits_\pi U_{p \pi} a_\pi + V^*_{p \pi} a_\pi^\dag, \quad c_n = \sum \limits_\nu U_{n \nu}^* a_\nu^\dag + V_{n \nu} a_\nu,
\end{equation}
where $\pi$$(\nu)$ denotes the quasi-proton(neutron) states, $a^\dag_{p(n)}, a_{p(n)}$ are the corresponding quasiparticle operators and $U$, $V$ are the Bogoliubov matrices. The external field operator $\hat{F}$ in the quasiparticle basis has the form as in Eq.~\eqref{eq:ftqrpa_field_op} with components
\begin{align}
\begin{split}
F^{11}_{\pi \nu} &= (U^\dag F U)_{\pi \nu}, \quad (F^\dag)^{11}_{\pi \nu} = (U^T F U^*)_{\pi \nu}, \\  
F^{20}_{\pi \nu} &= (U^\dag F V^*)_{\pi \nu}, \quad (F^\dag)^{20}_{\pi \nu} = (U^T F^\dag V)_{\pi \nu}, \\  
F^{02}_{\pi \nu} &= (V^T F U)_{\pi \nu}, \quad (F^\dag)^{02}_{\pi \nu} = (V^\dag F^\dag U^*)_{\pi \nu}, \\  
F^{\bar{11}}_{\pi \nu} &= (V^T F V^*)_{\pi \nu}, \quad (F^\dag)^{\bar{11}}_{\pi \nu} = (V^\dag F^\dag V)_{\pi \nu}. \\  
\end{split}
\end{align}
We can now evaluate the expression in Eq.~\eqref{eq:thermal_avg_of_comm} to get
\begin{align}
\begin{split}
\langle [ \hat{F}^\dag, \hat{F} ] \rangle &= \sum \limits_{pn} F_{pn} (F^\dag)_{pn} \left\{ (U f U^\dag)_{nn} \left[ (U^* U^T)_{pp} + (V V^\dag)_{pp} \right] \right. \\
&- (U^* f U^T)_{pp} \left[ (U U^\dag)_{nn} + (V^* V^T)_{nn} \right]    \\
&- (V^* f V^T)_{nn} \left[ (U^* U^T)_{pp} + (V V^\dag)_{pp} \right] \\
&+ (V f V^\dag)_{pp} \left[ (U U^\dag)_{nn} + (V^* V^T)_{nn} \right] \\
&+ \left. (U^* U^T)_{pp} (V^* V^T)_{nn} - (V V^\dag)_{pp} (U U^\dag)_{nn} \right\}.
\end{split}
\end{align}
Using the unitarity of the Bogoliubov transformation \cite{Ring2004}
\begin{equation}
U U^\dag + V^* V^T = 1, \quad U^* U^T + V V^\dag = 1,
\end{equation}
we can rewrite above expression as
\begin{align}
\begin{split}
\langle [ \hat{F}^\dag, \hat{F} ] \rangle &= \sum \limits_{pn} F_{pn} (F^\dag)_{pn} \left\{ (U f U^\dag)_{nn} + [V^* (1-f) V^T]_{nn} \right. \\
&- \left. (U^* f U^T)_{pp} - [V (1-f) V^\dag]_{pp} \right\} \\
&= \sum \limits_{pn} F_{pn} (F^\dag)_{pn} \left\{ n_n - n_p \right\},
\end{split}
\end{align}
where $n_n =  (U f U^\dag)_{nn} + [V^* (1-f) V^T]_{nn}$ and $n_p = (U^* f U^T)_{pp} + [V (1-f) V^\dag]_{pp}$ denote neutron and proton number of a given single-particle state. By inserting the GT operator as the external field
\begin{equation}
F_{pn} = \langle p || \boldsymbol{\sigma} \tau_- || n \rangle, \quad (F^\dag)_{pn} = \langle n || \boldsymbol{\sigma} \tau_+ || p \rangle,
\end{equation}
we get the well known result
\begin{equation}
\langle [ \hat{F}^\dag, \hat{F} ] \rangle = 3(N-Z).
\end{equation}
Same result can be reproduced within the H(F)-BCS formalism by replacing the Bogoliubov transformation with
\begin{equation}
c_p^\dag = \delta_{p \pi} (u_\pi a_\pi^\dag + v_\pi a_{\bar{\pi}}), \quad c_n = \delta_{n \nu} (u_\nu a_\nu + v_\nu a_{\bar{\nu}}^\dag),
\end{equation}
where $u,v$ are the H(F)-BCS amplitudes, and $\bar{\pi}$($\bar{\nu}$) denotes the time-reversed quasiparticle states~\cite{Ring2004}.

\section{Derivation of the Ikeda sum rule within the FT-PNQRPA}\label{sec:appc}

In this section we show that expression in Eq.~\eqref{eq:ft_qrpa_ikeda_sum_rule} reduces to the well-known result of the Ikeda sum rule, $3(N-Z)$~\cite{IKEDA1963271}. We can rewrite Eq.~\eqref{eq:ft_qrpa_ikeda_sum_rule} in matrix form as
\begin{align}\label{eq:ikeda_sum_appc}
\begin{split}
&\Sigma^0(\hat{F}) - \Sigma^0(\hat{F}^\dag) = \int \limits_{-\infty}^\infty d \omega \left[ \frac{dB}{d \omega}(\hat{F}) - \frac{d B}{d \omega} (\hat{F}^\dag) \right] \\
&= \sum \limits_{k\pm>0} \left( \frac{|\langle [\Gamma^k, \hat{F}] \rangle|^2}{1-e^{-\beta \Omega_k}} -  \frac{|\langle [\Gamma^k, \hat{F}^\dag] \rangle|^2}{1-e^{\beta \Omega_k}} \right) \\
&\qquad\qquad- \left( \frac{|\langle [\Gamma^k, \hat{F}^\dag] \rangle|^2}{1-e^{-\beta \Omega_k}} -  \frac{|\langle [\Gamma^k, \hat{F}] \rangle|^2}{1-e^{\beta \Omega_k}} \right) \\
&= \sum \limits_{k\pm>0} \left( |\langle [\Gamma^k, \hat{F}]\rangle|^2 - |\langle [\Gamma^k, \hat{F}^\dag]\rangle|^2 \right) \\
&= \begin{pmatrix}
\langle [\Gamma_1, \hat{F}] \rangle^* \ldots \langle [\Gamma_{2N}, \hat{F}] \rangle^* \langle [\Gamma_1, \hat{F}^\dag] \rangle \ldots \langle [\Gamma_{2N}, \hat{F}^\dag] \rangle
\end{pmatrix} \\
& \times \begin{pmatrix}
1 &  &  &  &  & \\
  & \ddots &  & & & \\
  &  &  1  &  & &  \\
    &  &    & -1  & &  \\
      &  &    &  & \ddots &  \\
        &  &    &  & & -1  \\
\end{pmatrix} \begin{pmatrix}
\langle [\Gamma_1, \hat{F} ] \rangle \\
\vdots \\
\langle [\Gamma_{2N}, \hat{F} ] \rangle \\
\langle [\Gamma_1, \hat{F}^\dag ] \rangle^* \\
\vdots \\
\langle [\Gamma_{2N}, \hat{F}^\dag ] \rangle^* \\
\end{pmatrix},
\end{split}
\end{align}
where $N$ is the number of (quasi)particle-(quasi)hole pairs, and the total dimension of the FT-PNQRPA matrix is $4N \times 4N$. The FT-PNQRPA equation can be written in the matrix form as already introduced in Sec.~\ref{sec:props}
\begin{equation}
[ \mathcal{E} + T \mathcal{H}] T \mathcal{X} = M {T} \mathcal{X} \mathcal{O},
\end{equation}
where $\mathcal{X}$ is the matrix whose columns consist of the temperature-independent FT-PNQRPA eigenvectors
\begin{equation}
\mathcal{X} = \begin{pmatrix}
P_1 \ldots P_{2N} & Q_1^* \ldots Q_{2N} \\
X_1 \ldots X_{2N} & Y_1^* \ldots Y_{2N} \\
Y_1 \ldots Y_{2N} & X_1^* \ldots X_{2N} \\
Q_1 \ldots Q_{2N} & P_1^* \ldots P_{2N} \\
\end{pmatrix},
\end{equation}
while $\mathcal{O} = \text{diag}(\Omega_1 \ldots \Omega_N, -\Omega_1, \ldots -\Omega_N)$ contains the FT-PNQRPA eigenvalues (cf.~Sec.~\ref{sec:props}). Starting from the FT-PNQRPA phonon operator in Eq.~\eqref{eq:ft_phonon_operator} we can evaluate the required ensemble averages according to Eq.~\eqref{eq:F_ensemble_avg}.
% \begin{align}
% \begin{split}
% \langle [\Gamma_k, \hat{F}]\rangle &= \sum \limits_{\pi \nu} P^{k *}_{\pi \nu} F^{11}_{\pi \nu} (f_\nu - f_\pi) + X^{k *}_{\pi \nu} F^{20}_{\pi \nu}(1-f_\pi - f_\nu) \\
% &+ Y^{k*}_{\pi \nu} F^{02}_{\pi \nu}(1-f_\pi - f_\nu) + Q^{k*}_{\pi \nu} F^{\bar{11}}_{\pi \nu} (f_\nu - f_\pi),
% \end{split} \label{eq:ensamble_averages_1_appc} \\
% \begin{split}
% \langle [\hat{F}^\dag, \Gamma_k]\rangle &= \sum \limits_{\pi \nu} P^{k *}_{\pi \nu} (F^\dag)^{\bar{11}}_{\pi \nu} (f_\nu - f_\pi) + X^{k *}_{\pi \nu} (F^\dag)^{02}_{\pi \nu}(1-f_\pi - f_\nu) \\
% &+ Y^{k*}_{\pi \nu} (F^\dag)^{20}_{\pi \nu}(1-f_\pi - f_\nu) + Q^{k*}_{\pi \nu} (F^\dag)^{11}_{\pi \nu} (f_\nu - f_\pi). \label{eq:ensamble_averages_2_appc}
% \end{split}
% \end{align}
% Using the above definitions 
It is straightforward to verify that
\begin{equation}
\begin{pmatrix}
\langle [\Gamma_1, \hat{F} ] \rangle \\
\vdots \\
\langle [\Gamma_{2N}, \hat{F} ] \rangle \\
\langle [\Gamma_1, \hat{F}^\dag ] \rangle^* \\
\vdots \\
\langle [\Gamma_{2N}, \hat{F}^\dag ] \rangle^* \\
\end{pmatrix} = \mathcal{X}^\dag T \mathcal{F}.
\end{equation}
Using the normalization condition for the FT-PNQRPA eigenvectors
\begin{equation}
\mathcal{X}^\dag T M \mathcal{X} = M,
\end{equation}
Eq.~\eqref{eq:ikeda_sum_appc} can be written as
\begin{align}
\begin{split}
&\Sigma^0(\hat{F}) - \Sigma^0(\hat{F}^\dag) = \mathcal{F}^\dag T \mathcal{X} M \mathcal{X}^\dag T \mathcal{F} = \mathcal{F}^\dag T M \mathcal{F} \\
&=  F^{11}_{\pi \nu} (F^\dag)^{11}_{\pi \nu}(f_\nu - f_\pi) + (1-f_\nu - f_\pi) F^{20}_{\pi \nu} (F^\dag)^{20}_{\pi \nu} \\
&- (1-f_\pi -f_\nu) F^{02}_{\pi \nu} (F^\dag)^{02}_{\pi \nu} - F^{\bar{11}}_{\pi \nu} (F^\dag)^{\bar{11}}_{\pi \nu} (f_\nu - f_\pi),
\end{split}
\end{align}
which agrees with the expression for $ \langle [ \hat{F}^\dag, \hat{F} ] \rangle$ of Eq.~(\ref{eq:thermal_avg_of_comm}).

\clearpage
\newpage
\bibliographystyle{apsrev4-1}
%\bibliography{references}% Produces the bibliography via BibTeX.

\end{document}